\documentclass[twoside]{IEEEtran}

\usepackage{amsmath,amsfonts,amssymb,mathrsfs}
\usepackage{algorithmicx}
\usepackage{algorithm}
\usepackage{array}
\usepackage[caption=false,font=normalsize,labelfont=sf,textfont=sf]{subfig}
\usepackage{textcomp}
\usepackage{stfloats}
\usepackage{url}
\usepackage{verbatim}
\usepackage{graphicx}
\usepackage{balance}
\usepackage{cite}

\usepackage{ifxetex}
\ifxetex
\usepackage{fontspec}
\else
\usepackage[T1]{fontenc}
\fi

\newenvironment{proof}{\begin{IEEEproof}}{\end{IEEEproof}} 
\newtheorem{lemma}{Lemma}
\newtheorem{corollary}[lemma]{Corollary}
\newtheorem{theorem}[lemma]{Theorem}
\newtheorem{definition}{Definition}
\newtheorem{example}{Example}

\usepackage{algpseudocode} 
\usepackage[inline]{enumitem} 
\usepackage{physics} 
\usepackage{xfp}

\usepackage{tikz}
\usepackage[compat=0.6]{yquant}
\usetikzlibrary{3d, patterns}

\yquantset{
    operators/every control box/.append style={shape=yquant-circle, draw}
}
\makeatletter
\yquant@langhelper@declare@command@uncontrolled%
{controlbox}%
\yquant@register@get@allowmultitrue%
{%
    \expandafter\yquant@prepare@injection%
    \expandafter{\yquant@lang@attr@value}%
    {/yquant/operators/every control box}%
}
\makeatother

\mathchardef\mathhyphen="2D 

\usepackage{orcidlink} 
\hypersetup{hidelinks}

\title{Almost Optimal Synthesis of Reversible Function in Qudit Model}
\author{
    Buji~Xu\textsuperscript{\orcidlink{0009-0008-0031-9106}},
    Junhong~Nie\textsuperscript{\orcidlink{0000-0001-9427-9329}},
    and Xiaoming~Sun\textsuperscript{\orcidlink{0000-0002-0281-1670}}

    \thanks{This work was supported in part by the National Natural Science Foundation of China Grants No. 62325210, 92465202, and the Strategic Priority Research Program of Chinese Academy of Sciences Grant No. XDB28000000.

        Buji~Xu, Junhong~Nie and Xiaoming~Sun are with State Key Lab of Processors, Institute of Computing Technology, Chinese Academy of Sciences, Beijing 100190, China, and also with School of Computer Science and Technology, University of Chinese Academy of Sciences, Beijing 100049, China (email: xubuji20@mails.ucas.ac.cn; niejunhong19z@ict.ac.cn; sunxiaoming@ict.ac.cn).
}}

\begin{document}

    \maketitle

    \begin{abstract}
        Quantum oracles are widely adopted in problems, like query oracle in Grover's algorithm, cipher in quantum cryptanalytic and data encoder in quantum machine learning.
        Notably, the bit-flip oracle, capable of flipping the state based on a given classical function, emerges as a fundamental component in the design and construction of quantum algorithms.
        Devising methods to optimally implement the bit-flip oracle essentially translates to the efficient synthesis of reversible functions.
        Prior research has primarily focused on the qubit model, leaving the higher dimensional systems, i.e. qudit model, largely unexplored.
        By allowing more than two computational bases, qudit model can fully utilize the multi-level nature of the underlying physical mechanism.
        We propose a method to synthesize even permutations in $A_{d^{n}}$ using $\Theta(d)$ $(n - 1)$-qudit sub-circuits, which achieve asymptotic optimality in the count of sub-circuits.
        Moreover, we introduce a technique for synthesizing reversible functions employing $O\left( n d^{n} \right)$ two-qudit gates and only a single ancilla.
        This is asymptotically tight in terms of $d$ and asymptotically almost tight in terms of $n$.
    \end{abstract}


    \section{Introduction} \label{section:intro}

    \IEEEPARstart{R}{eversible} computing plays an important role in fields including cryptography~\cite{shiBitPermutationInstructions2000}, signal processing~\cite{egnerDecomposingPermutationConjugated1997}.
    In reversible computing, a fundamental problem is how to decompose complex reversible functions into simpler ones.
    Shannon~\cite{shannonSynthesisTwoterminalSwitching1949} introduced the concept of circuit complexity and proved that almost all boolean function on $n$ variables require circuits of size $\Theta \left( \frac{2^n}{n} \right)$.
    In the paper, we will study an extended object --- reversible function on $n$ $d$-ary variables --- in the context of quantum computing.

    Quantum computing is a novel computational model that harnesses the principles of quantum mechanics. It has demonstrated potential advantages in solving problems such as integer factorization~\cite{shorAlgorithmsQuantumComputation1994} and searching~\cite{groverFastQuantumMechanical1996,longGroverAlgorithmZero2001,heQuantumSearchPrior2024}.
    The superiority of quantum algorithms is often highlighted within the framework of query complexity.
    For example, Simon's problem~\cite{simonPowerQuantumComputation1994} exhibits an exponential separation between quantum and classical query complexities, with a polynomial number of queries in the quantum setting compared to the exponential number needed in classical setting.
    To run these algorithms on real hardware, we need to give an implementation of the oracles in quantum circuit, which is essential to quantum information processing as all possible transformations between quantum channels can be implemented by quantum circuits with open slots~\cite{chiribellaQuantumCircuitArchitecture2008}.
    Our focus lies on synthesizing the reversible counterparts of classical functions.
    In some literature, this may be called as permutation oracle, $O_{\pi} \ket{x} = \ket{\pi(x)}$.
    For example, block ciphers in symmetric-key schemes, like AES, are essentially complex permutations.
    Constructing resource-efficient quantum circuits for these primitives has become an active research area, as NIST used the circuit complexity of AES and SHA-3 as baseline to estimate security levels in its post-quantum cryptography standardization effort.
    S-box is the non-linear component of AES and the primary source of difficulty in implementing AES on quantum computers.
    When the number of qubits are restricted, it is unlikely to utilize the algebraic structure of S-box and hence researchers may turn to general methods.
    For instance, a recent study by Huang~et~al.~\cite{huangConstructingQuantumImplementations2025} explored how to optimize the design of AES S-box with minimal number of qubits using SAT solvers.
    For another example, the hidden subgroup problem of symmetric group is especially important in the theory of quantum computing, as an efficient quantum algorithm for it implies an efficient quantum algorithm for graph isomorphism~\cite{ettingerQuantumObservableGraph1999}.
    Though it remains an open challenge, an efficient implementation of permutation oracle might be helpful.
    We will delve into the investigation of the synthesis of reversible (classical) functions.

    Previous research in quantum computing has predominantly centered around the qubit model, while the qudit~\cite{wangQuditsHighdimensionalQuantum2020} model remains largely unexplored.
    Each qudit has $d$ basic states instead of 2, enabling the utilization of the rich multilevel structure of the underlying physical hardware~\cite{ringbauerUniversalQuditQuantum2022}.
    Pioneer implementations of qudit systems~\cite{chiProgrammableQuditbasedQuantum2022,liuPerformingmathrmSUdOperations2023,ringbauerUniversalQuditQuantum2022} and progress of error mitigation~\cite{omanakuttanFaulttolerantQuantumComputation2024,gossExtendingComputationalReach2024} demonstrate promising possibilities for practical use in the future.
    Theoretical investigations into qudit systems have also gained attention.
    For example, studies like existence of distinguishable bases in three-dimensional subspaces of qutrit-qudit system under one-way LPCC~\cite{songExistenceDistinguishableBases2024} provide valuable insights.
    Moreover, the development of techniques for optimizing qudit circuits is rapidly advancing, which finds rich applications in quantum algorithm design.
    For example, Baker~et~al.~\cite{bakerImprovedQuantumCircuits2020} suggested an exponential improvement over qubit in terms of circuit depth is possible.

    We study the synthesis of reversible function in two aspects.
    One is to decompose reversible function on $n$ variables to that on $(n - 1)$ variables.
    This corresponds to building quantum circuit using smaller sub-circuits.
    The other one is to decompose reversible function on $n$ variables to that on $2$ variables.
    This corresponds to building quantum circuit using $2$-qudit gates, which is universal as pointed out by Brylinskis~\cite{brylinskiUniversalQuantumGates2002}.

    For the sub-circuits count, Jiang~et~al.~\cite{jiangStructuredDecompositionReversible2020} gave a finite alternation result for reversible boolean circuits.
    They proved that even permutation on $2^n$ elements can be decomposed to $O(1)$ permutation on $2^{n - 1}$ elements.
    We will prove that $\Theta(d)$ quantum sub-circuits of $(n - 1)$ qudits can implement an even permutation on $d^n$ elements.

    For the circuit size, Wu and Li~\cite{wuAsymptoticallyOptimalSynthesis2024} gave a $\Theta \left( n 2^{n} \frac{1}{\log n} \right)$ bound for qubit model.
    It uses optimized Gaussian elimination to speed up the procedure~\cite{patelOptimalSynthesisLinear2008}.
    The extension to qudit model requires $d$ to be prime as modular inverse is used in Gaussian elimination.
    Besides, a $O(d^{2})$ overhead occurs as the row operation needs $O(d^{2})$ gates, which is not counted when $d = 2$.
    For general $d$, a translation from qubit circuits to qudit circuits can be performed, which needs $O(d^2)$ overhead. (For example, $CNOT$ gate in qubit circuits may be translated into $(\ket{0} \mathhyphen X_{02}) (\ket{0} \mathhyphen X_{13}) (\ket{2} \mathhyphen X_{02}) (\ket{2} \mathhyphen X_{13})$.)
    This suggests a $O \left( d^{2} n' 2^{n'} \frac{1}{\log n'} \right) = O \left( n d^{n} \frac{d^2 \log d}{\log \log d + \log n} \right)$ bound for qudit model, where $n' = n \log d$.

    As for qudit model, Zi~et~al.~\cite{ziOptimalSynthesisMulticontrolled2023} gave a $O(n d^{n + 3})$ bound for qudit model with assitance of $1$ ancilla qudit.
    If $O(n)$ ancilla qudits are allowed, Roy~et~al.~\cite{royQuditZHcalculusGeneralised2023} proved a $O \left( n d^{n} \right)$ upper bound when $d$ is odd.
    We will propose a new synthesis algorithm, which achieves $O \left( n d^{n} \right)$ quantum circuit size for general $d$ and uses at most $1$ ancilla qudit.
    The circuit size is optimal up to a constant factor when $d = \Omega(n)$ and is optimal up to a logarithm factor when $d = o(n)$.

    The structure of the rest part is as follows.
    In section~\ref{section:prelim}, we will introduce symbols used and some background knowledge.
    In section~\ref{section:result01}, we will show how to synthesize reversible function using $(n - 1)$-qudit sub-circuits.
    In section~\ref{section:result02}, we will show how to synthesize reversible function using two-qudit gates.
    In section~\ref{section:concl}, we will conclude this paper.

    \section{Preliminary} \label{section:prelim}

    Qudit is a computational unit analogous to the qubit, whose state
    \begin{equation*}
        \ket{\psi} = \sum_{x = 0}^{d - 1} \alpha_{x} \ket{x}, \quad \sum_{x = 0}^{d - 1} |\alpha_{x}|^2 = 1, \alpha_{x} \in \mathbb{C}
    \end{equation*}
    is a unit vector in the $d$-dimensional Hilbert space spanned by orthonormal basis $\{ \ket{0}, \dots, \ket{d - 1} \}$.
    We will use $d$ to represent the dimension of qudit hereafter.
    For $v = \sum_{k = 1}^{m} v_{k} d^{k - 1}$, $\ket{\mathbf{v}}$ (sometimes $\ket{v}$) denotes $\ket{v_{m}} \otimes \dots \otimes \ket{v_{2}} \otimes \ket{v_{1}} = \ket{v_{m} \dots v_{2} v_{1}}$.
    $\mathbf{v}_{i:j}$ will be used to denote the subvector $(v_{i}, v_{i + 1}, \dots, v_{j})$.

    The following gates are called $X$ gate
    \begin{equation*}
        X_{x, y} = \ket{x}\bra{y} + \ket{y}\bra{x} + \sum_{z \neq x, y} \ket{z}\bra{z}
    \end{equation*}
    and the following gates are called $X^{(m)}$ gate (which acts on multiple qudits)
    \begin{equation*}
        X_{\mathbf{x}, \mathbf{y}} = \ket{\mathbf{x}}\bra{\mathbf{y}} + \ket{\mathbf{y}}\bra{\mathbf{x}} + \sum_{\mathbf{z} \neq \mathbf{x}, \mathbf{y}} \ket{\mathbf{z}}\bra{\mathbf{z}}
    \end{equation*}
    where $|\mathbf{x}| = m$.

    The following gates are called controlled unitary
    \begin{equation*}
        \ket{x} \mathhyphen U = \ket{x} \otimes U + \sum_{y \neq x} \ket{y} \otimes I ,
    \end{equation*}
    which applies unitary transformation $U$ to target qudit(s) if the control qudit is in state $\ket{x}$. If $U = \ket{x'} \mathhyphen U'$ is a controlled unitary, $\ket{x} \mathhyphen U = \ket{x x'} \mathhyphen U'$ is called a controlled controlled-unitary. And $\ket{\mathbf{x}} \mathhyphen U$ is called $C^{m}U$ gate where $|\mathbf{x}| = m$.

    If $f : \{0, 1, \cdots, d - 1\}^n \to \{0, 1, \cdots, d - 1\}^n$ is bijective, it is called a $n$-dit reversible function.
    Such function can be viewed as a permutation in the symmetric group $S_{d^n}$.
    Permutations are commonly expressed in cyclic notation.
    \begin{definition}
        A $k$-cycle $(a_{1} \ a_{2} \ \dots \ a_{k})$ maps $a_{i}$ to $a_{i + 1}$ ($1 \leq i < k$) and maps $a_{k}$ to $a_{1}$.
    \end{definition}
    The composition of permutation adopts the rule $(f \circ g)(x) = f(g(x))$.
    For example, given $f = (2 \ 3)$, $g = (1 \ 2 \ 3)$, $(fg)(3) = f(g(3)) = f(1) = 1$ while $(gf)(3) = g(f(3)) = g(2) = 3$.
    A permutation is even if it can be decomposed into even number of 2-cycles.
    For example, $(1 \ 2 \ 3)$ is even as $(1 \ 2 \ 3) = (1 \ 3)(1 \ 2)$.
    The alternating group $A_{m}$ is the subgroup of $S_{m}$ consisting of all even permutations.

    \section{Decomposition into \texorpdfstring{$(n - 1)$}{(n-1)}-qudit blocks} \label{section:result01}

    In this section, we will explore the decomposition of reversible functions in $A_{d^n}$ using those in $S_{d^{n-1}}$.
    This is essentially realizing reversible functions in $A_{d^n}$ using $(n-1)$-qudit reversible quantum circuits.
    $A_{d^n}$ is chosen because it is impossible to implement $\pi \in S_{d^{n}} \setminus A_{d^{n}}$ using $S_{d^{n - 1}}$ when $d$ is even.
    We will first concentrate on the case of $n = 3$.
    Then we will extend the techniques to handle arbitrary value of $n$.
    Finally, we will demonstrate that our sub-circuit count is asymptotically optimal.

    \subsection{Decomposition for the Case \texorpdfstring{$n = 3$}{n=3}} \label{section:result01:cube}

    When $n = 3$, we can visualize the state space of $3$-qudit system as a $d \times d \times d$ cube. A $2$-qudit reversible quantum circuit can be restated as
    \begin{itemize}
        \item Select a plane $P$ and an arbitrary permutation $\pi$ acting on $P$, i.e. $\pi \in S_{d \times d}$.
        Apply $\pi$ to all planes parallel to $P$, including $P$ itself.
    \end{itemize}
    This is called a basic operation as it corresponds to one $(n-1)$-qudit reversible quantum circuit.

    \begin{example}
        In \figurename~\ref{fig:3d_visual}, we choose the xy-plane as $P$ and the counterclockwise rotation by $45$\textdegree~around the center as $\pi$. We need to apply $\pi$ for each plane in the $z$-direction. As a result, \figurename~\ref{fig:3d_visual:a} is transformed into \figurename~\ref{fig:3d_visual:b}.
    \end{example}

    \begin{figure*}
        \centering
        \subfloat[\label{fig:3d_visual:a}]{
            \begin{tikzpicture}
                \begin{scope}[canvas is xy plane at z=3]
                    \draw [step=1,help lines] (0,0) grid (3,3);
                    \pattern [pattern=horizontal lines] (1,0) rectangle (2,3);
                    \pattern [pattern=dots] (2,0) rectangle (3,3);
                \end{scope}
                \begin{scope}[canvas is yz plane at x=3]
                    \draw [step=1,help lines] (0,0) grid (3,3);
                    \pattern [pattern=dots] (0,2) rectangle (3,3);
                    \pattern [pattern=north east lines] (0,1) rectangle (3,2);
                \end{scope}
                \begin{scope}[canvas is xz plane at y=3]
                    \draw [step=1,help lines] (0,0) grid (3,3);
                    \draw[thick] (1.5, 1.5) ellipse (1.25 and 1);
                    \draw[->, thick] (2.06, 2.43) -- (2.07, 2.42) node[above] {$\pi$};
                    \node[circle, fill, inner sep=1pt] at (1.5, 2.5) {};
                    \node[circle, fill, inner sep=1pt] at (2.375, 2.25) {};
                \end{scope}
                \coordinate (S) at (-2.5,0,0);
                \node (Z) at (-2.5,0,1.4) {$x$};
                \node (X) at (-1.5,0,0) {$y$};
                \node (Y) at (-2.5,1,0) {$z$};
                \draw[->] (S) -- (X);
                \draw[->] (S) -- (Y);
                \draw[->] (S) -- (Z);
            \end{tikzpicture}
        }
        \subfloat[\label{fig:3d_visual:b}]{
            \begin{tikzpicture}
                \begin{scope}[canvas is xy plane at z=3]
                    \draw [step=1,help lines] (0,0) grid (3,3);
                    \pattern [pattern=horizontal lines] (2,0) rectangle (3,3);
                \end{scope}
                \begin{scope}[canvas is yz plane at x=3]
                    \draw [step=1,help lines] (0,0) grid (3,3);
                    \pattern [pattern=horizontal lines] (0,2) rectangle (3,3);
                    \pattern [pattern=dots] (0,1) rectangle (3,2);
                    \pattern [pattern=north east lines] (0,0) rectangle (3,1);
                \end{scope}
                \begin{scope}[canvas is xz plane at y=3]
                    \draw [step=1,help lines] (0,0) grid (3,3);
                \end{scope}
            \end{tikzpicture}
        }
        \caption{Visualization of the state space of $3$-qudit system}
        \label{fig:3d_visual}
    \end{figure*}

    \subsubsection{Technical overview}

    \begin{figure}
        \centering
        \begin{tikzpicture}
            \node (lem2) at (0,0) {Lemma~\ref{thm:lem01weak}};
            \node (thm1) at (0,-1.2) {Theorem~\ref{thm:thm01}};
            \node[style={align=center}] (lem3) at (2,1) {Lemma~\ref{thm:lem:op_[f,g]} \\ ($\pi \sigma \pi^{-1} \sigma^{-1}$)};
            \node[style={align=center}] (lem4) at (3,-0.2) {Lemma~\ref{thm:lem:[f,g]} \\ (even permutation for edge)};
            \node (lem5) at (4,1) {Lemma~\ref{thm:lem:decomp_commut_cycle}};
            \node[style={align=center}] (lem6) at (6.5,0) {Lemma~\ref{thm:lem:edges_on_plane} \\ ($\sigma_{i}^{-1} \sigma_{\pi(i)}$)};
            \node[style={align=center}] (cor7) at (3,-1.5) {Corollary~\ref{thm:cor:op1_plane} \\ (type-1 op for plane)};
            \node[style={align=center}] (cor8) at (6.5,-1.5) {Corollary~\ref{thm:cor:op2_plane} \\ (type-2 op for plane)};
            \node[style={align=center}] (cor9) at (3,-3) {Corollary~\ref{thm:cor:plane_even} \\ (even permutation for plane)};
            \node[style={align=center}] (cor10) at (3,-4.5) {Theorem~\ref{thm:cor:cube_even} \\ (even permutation for cube)};

            \draw[->] (lem2) -- (thm1);
            \draw[->] (lem3) -- (lem4);
            \draw[->] (lem5) -- (lem4);
            \draw[->] (lem4) -- (cor7);
            \draw[->] (lem6) -- (cor7);
            \draw[->] (lem6) -- (cor8);
            \draw[->] (thm1) -- (cor9);
            \draw[->] (cor7) -- (cor9);
            \draw[->] (cor8) -- (cor9);
            \draw[->] (thm1) to[bend right] (cor10);
            \draw[->] (cor9) -- node[right] {$O(d)$} (cor10);
        \end{tikzpicture}
        \caption{Overview for Section \ref{section:result01:cube}}
        \label{fig:result01overview}
    \end{figure}

    As shown by \figurename~\ref{fig:result01overview}, the main technical theorem of this subsection is Theorem~\ref{thm:thm01}, which is inspired by the following theorem given by Sun~et~al.~\cite{sunGeneralizedShuffleexchangeProblem2022}.

    \begin{theorem}[\cite{sunGeneralizedShuffleexchangeProblem2022}]
        \label{thm:thm_rcr}
        A permutation $\pi \in S_{s \times t}$ is of type-R (or \textit{type-C}) if it does not change the row (or \textit{column}) of any element.
        An arbitrary permutation $\pi \in S_{s \times t}$ can be decomposed into $3$ permutations $\pi_{3} \pi_{2} \pi_{1}$ where $\pi_{1}$ and $\pi_{3}$ are of type-R while $\pi_{2}$ is of type-C.
    \end{theorem}

    We prove that it can be extended by restricting permutations to be even.

    \begin{theorem}[informal version of Theorem~\ref{thm:thm01}]
        An arbitrary permutation $\pi \in A_{s \times t}$ can be decomposed into $O(1)$ permutations where each permutation is an even permutation of type-R or type-C.
    \end{theorem}

    Our goal is to prove Theorem~\ref{thm:cor:cube_even} and the core of the proof is Theorem~\ref{thm:thm01}.
    Lemma~\ref{thm:lem:op_[f,g]} and Lemma~\ref{thm:lem:edges_on_plane} construct two different ways of utilizing the basic operation.
    Lemma~\ref{thm:lem:op_[f,g]} can perform permutations of form $\pi \sigma \pi^{-1} \sigma^{-1}$ on an edge of the cube.
    It can be improved to any even permutation on an edge (Lemma~\ref{thm:lem:[f,g]}) with the help of Lemma~\ref{thm:lem:decomp_commut_cycle}.
    Lemma~\ref{thm:lem:edges_on_plane} can perform $\sigma_{i}^{-1} \sigma_{\pi(i)}$ to $i$-th row of a plane of the cube.
    This can be used to perform $\tau_{i}$ to $i$-th row of a plane, where $\tau_{i}$ is arbitrary even permutation when $i < d$ and $\tau_{d}$ can be decided by $\tau_{1}, \dots, \tau_{d - 1}$.
    Combining these results, we can provide the two different kinds of even operations (Corollary~\ref{thm:cor:op1_plane} and Corollary~\ref{thm:cor:op2_plane}) needed by Theorem~\ref{thm:thm01} and hence achieve arbitrary even permutation on a plane of the cube (Corollary~\ref{thm:cor:plane_even}).
    By flattening a $d \times d \times d$ cube as a $d \times d^{2}$ plane, we reduce it to a 2-dimensional problem.
    We can apply Theorem~\ref{thm:thm01} again to obtain the goal (Theorem~\ref{thm:cor:cube_even}).

    \subsubsection{Main technical lemma}

    The 2 types of permutations needed by Theorem~\ref{thm:thm01} are formally defined as follows.

    \begin{definition}[type-1 permutation]
        A permutation $\pi \in A_{s \times t}$ is of type-1r if it can be decomposed into $s$ permutations $\pi_{1}, \dots, \pi_{s} \in A_{t}$ where $\pi_{i}$ only acts on row $i$ ($i = 1, \dots, s$).
        A permutation $\sigma \in A_{s \times t}$ is of type-1c if it can be decomposed into $t$ permutations $\sigma_{1}, \dots, \sigma_{t} \in A_{s}$ where $\sigma_{i}$ only acts on column $i$ ($i = 1, \dots, t$).
    \end{definition}

    \begin{definition}[type-2 permutation]
        A permutation $\pi \in A_{s \times t}$ is of type-2r if $\pi = \prod_{i = 1}^{2k} ((r_{i}, c_{1}) \ (r_{i}, c_{2}))$ where $k \leq t / 2$, $c_{1} \neq c_{2}$ and $r_{i_{1}} \neq r_{i_{2}}$ when $i_{1} \neq i_{2}$.
        A permutation $\pi \in A_{s \times t}$ is of type-2c if it $\sigma = \prod_{i = 1}^{2k} ((r_{1}, c_{i}) \ (r_{2}, c_{i}))$ where $k \leq s / 2$, $r_{1} \neq r_{2}$ and $c_{i_{1}} \neq c_{i_{2}}$ when $i_{1} \neq i_{2}$.
    \end{definition}

    The proof of Theorem~\ref{thm:thm01} depends on the following Lemma, which will be proved in section~\ref{section:result01:lem01}.

    \begin{lemma}
        \label{thm:lem01weak}
        $((r_{0}, c_{1}) \ (r_{0}, c_{2})) ((r_{1}, c_{0}) \ (r_{2}, c_{0}))$, which swaps the $c_{1}$-th and $c_{2}$-th element of row $r_{0}$ and swaps the $r_{1}$-th and $r_{2}$-th element of column $c_{0}$, can be written as product of $O(1)$ type-2 permutations when $s, t \geq 3$.
    \end{lemma}

    Now, we state and prove the main theorem of this subsection.

    \begin{theorem}
        \label{thm:thm01}
        An arbitrary permutation $\pi \in A_{s \times t}$ can be decomposed into $O(1)$ permutations where each permutation is of type-1r, type-1c, type-2r or type-2c.
    \end{theorem}

    \begin{proof}
        An observation is that an even permutation of type-R can be decomposed as a permutation of type-1r and a permutation of type-2r.
        Let
        \begin{equation*}
            \pi'_{i} = \begin{cases}
                \pi_{i}, & \pi_{i}~\text{is even}, \\
                (1 \ 2) \pi_{i}, & \text{otherwise} .
            \end{cases}
        \end{equation*}
        Permutation $\pi' = \pi'_{1} \dots \pi'_{s}$ is of type-1r.
        Then $\pi (\pi'^{-1})$ is of type-2r with $c_{1} = 1, c_{2} = 2$ and $\{ r_{1}, \dots, r_{2k} \} = \{ i : \pi_{i} \neq \pi'_{i} \}$.
        This is valid as $\pi (\pi')^{-1}$ must be even.
        Similarly, an even permutation of type-C can be decomposed as a permutation of type-1c and a permutation of type-2c.

        We will dive in the general case, an even permutation $\pi$ in $S_{s \times t}$. Consider its decomposition $\pi_{2} \sigma \pi_{1}$ given by Theorem~\ref{thm:thm_rcr}.
        \begin{itemize}
            \item $\pi_{1}, \pi_{2}$ and $\sigma$ are all even permutations. Simply use the observation.
            \item $\pi_{1}$ and $\sigma$ are odd permutations.
            Rewrite $\pi$ as
            \begin{equation*}
                \pi_{2} \cdot \sigma \sigma' \cdot (\sigma')^{-1} (\pi')^{-1} \cdot \pi' \pi_{1}
            \end{equation*}
            where $\sigma'$ swaps two elements in some row $r_{0}$ and $\pi'_{1}$ swaps two elements in some column $c_{0}$.
            Apply the observation for $\pi' \pi_{1}$, $\sigma \sigma'$ and $\pi_{2}$.
            And $(\sigma')^{-1}(\pi')^{-1}$ can be performed with $O(1)$ operations by Lemma~\ref{thm:lem01weak}.
            \item $\pi_{1}$ and $\pi_{2}$ are odd permutations.
            Similarly, rewrite $\pi$ as
            \begin{equation*}
                \pi_{2} \pi'' \cdot (\pi'')^{-1} (\sigma'')^{-1} \cdot \sigma'' \sigma \sigma' \cdot (\sigma')^{-1} (\pi')^{-1} \cdot \pi' \pi_{1}
            \end{equation*}
            where $\sigma'$ and $\sigma''$ additionally satisfy the condition $\sigma'' \sigma \sigma' = \sigma$. This condition can be achieved by choosing $\sigma' = (a_{1} \ a_{2}), \sigma'' = (\sigma(a_{1}) \ \sigma(a_{2}))$.
        \end{itemize}
    \end{proof}

    \subsubsection{Manipulating an edge of the cube}

    In this subsubsection, we show that the effect of operation can be restricted to an edge of the cube.
    The basic idea is to use a sequence of operations $\pi, \sigma, \pi^{-1}, \sigma^{-1}$.
    As $\pi$ cancels with $\pi^{-1}$, only those affected by both $\pi$ and $\sigma$ may be permuted.

    \begin{lemma}
        \label{thm:lem:op_[f,g]}
        Let $\pi$, $\sigma$ be permutations in $S_{d}$.
        It takes $O(1)$ basic operations to do permutation $\sigma^{-1} \pi^{-1} \sigma \pi$ on a $1 \times 1 \times d$ sub-cuboid.
    \end{lemma}

    \begin{proof}
        Without loss of generality, assume this sub-cuboid is the upper edge of the front face.
        As shown by \figurename~\ref{fig:[f,g]}, we first apply $\pi$ to each row of the upper face and then $\sigma$ to each row of the front face.
        Subsequently, we apply $\pi^{-1}$ to each row of the upper face and then $\sigma^{-1}$ to each row of the front face.
        If a cubie ($1 \times 1 \times 1$ miniature cube) is neither on the top face nor on the front face, it undergoes no transformation.
        For cubies solely on the top face, the successive application of $\pi$ and its inverse $\pi^{-1}$ results in the identity transformation.
        Similarly, cubies only on the front face remain unaffected.
        Cubies at the intersection of the top and front faces undergo the transformation $\sigma^{-1} \pi^{-1} \sigma \pi$.
    \end{proof}

    \begin{figure*}
        \centering
        \subfloat[\label{fig:[f,g]:a}]{
            \centering
            \begin{tikzpicture}
                \begin{scope}[canvas is xy plane at z=3]
                    \draw [step=1,help lines] (0,0) grid (3,3);
                    \pattern [pattern=dots] (0,0) rectangle (1,3);
                    \pattern [pattern=vertical lines] (1,0) rectangle (2,3);
                    \pattern [pattern=north east lines] (2,0) rectangle (3,3);
                    \draw[thick] (1.5, 2.5) ellipse (1.25 and 0.25);
                    \draw[->, thick] (1.49, 2.25) -- (1.51, 2.25) node[below=0.1cm, fill=white, inner sep=2pt] {$\pi$};
                \end{scope}
                \begin{scope}[canvas is yz plane at x=3]
                    \draw [step=1,help lines] (0,0) grid (3,3);
                \end{scope}
                \begin{scope}[canvas is xz plane at y=3]
                    \draw [step=1,help lines] (0,0) grid (3,3);
                \end{scope}
            \end{tikzpicture}
        }
        \subfloat[\label{fig:[f,g]:b}]{
            \centering
            \begin{tikzpicture}
                \begin{scope}[canvas is xy plane at z=3]
                    \draw [step=1,help lines] (0,0) grid (3,3);
                    \pattern [pattern=dots] (0,0) rectangle (1,2);
                    \pattern [pattern=vertical lines] (1,0) rectangle (2,2);
                    \pattern [pattern=north east lines] (2,0) rectangle (3,2);
                    \pattern [pattern=dots] (1,2) rectangle (2,3);
                    \pattern [pattern=vertical lines] (2,2) rectangle (3,3);
                    \pattern [pattern=north east lines] (0,2) rectangle (1,3);
                \end{scope}
                \begin{scope}[canvas is yz plane at x=3]
                    \draw [step=1,help lines] (0,0) grid (3,3);
                \end{scope}
                \begin{scope}[canvas is xz plane at y=3]
                    \draw [step=1,help lines] (0,0) grid (3,3);
                    \draw[thick] (1, 2.5) ellipse (0.75 and 0.25);
                    \draw[->, thick] (0.99, 2.25) -- (1.01, 2.25) node[above] {$\sigma$};
                \end{scope}
            \end{tikzpicture}
        }
        \subfloat[\label{fig:[f,g]:c}]{
            \centering
            \begin{tikzpicture}
                \begin{scope}[canvas is xy plane at z=3]
                    \draw [step=1,help lines] (0,0) grid (3,3);
                    \pattern [pattern=dots] (0,0) rectangle (1,2);
                    \pattern [pattern=vertical lines] (1,0) rectangle (2,2);
                    \pattern [pattern=north east lines] (2,0) rectangle (3,2);
                    \pattern [pattern=north east lines] (1,2) rectangle (2,3);
                    \pattern [pattern=vertical lines] (2,2) rectangle (3,3);
                    \pattern [pattern=dots] (0,2) rectangle (1,3);
                    \draw[thick] (1.5, 2.5) ellipse (1.25 and 0.25);
                    \draw[->, thick] (1.51, 2.25) -- (1.49, 2.25) node[below=0.1cm, fill=white, inner sep=2pt] {$\pi^{-1}$};
                \end{scope}
                \begin{scope}[canvas is yz plane at x=3]
                    \draw [step=1,help lines] (0,0) grid (3,3);
                \end{scope}
                \begin{scope}[canvas is xz plane at y=3]
                    \draw [step=1,help lines] (0,0) grid (3,3);
                \end{scope}
            \end{tikzpicture}
        }
        \\
        \subfloat[\label{fig:[f,g]:d}]{
            \centering
            \begin{tikzpicture}
                \begin{scope}[canvas is xy plane at z=3]
                    \draw [step=1,help lines] (0,0) grid (3,3);
                    \pattern [pattern=dots] (0,0) rectangle (1,2);
                    \pattern [pattern=vertical lines] (1,0) rectangle (2,2);
                    \pattern [pattern=north east lines] (2,0) rectangle (3,2);
                    \pattern [pattern=vertical lines] (1,2) rectangle (2,3);
                    \pattern [pattern=dots] (2,2) rectangle (3,3);
                    \pattern [pattern=north east lines] (0,2) rectangle (1,3);
                \end{scope}
                \begin{scope}[canvas is yz plane at x=3]
                    \draw [step=1,help lines] (0,0) grid (3,3);
                \end{scope}
                \begin{scope}[canvas is xz plane at y=3]
                    \draw [step=1,help lines] (0,0) grid (3,3);
                    \draw[thick] (1, 2.5) ellipse (0.75 and 0.25);
                    \draw[->, thick] (1.01, 2.25) -- (0.99, 2.25) node[above] {$\sigma^{-1}$};
                \end{scope}
            \end{tikzpicture}
        }
        \subfloat[\label{fig:[f,g]:e}]{
            \centering
            \begin{tikzpicture}
                \begin{scope}[canvas is xy plane at z=3]
                    \draw [step=1,help lines] (0,0) grid (3,3);
                    \pattern [pattern=dots] (0,0) rectangle (1,2);
                    \pattern [pattern=vertical lines] (1,0) rectangle (2,2);
                    \pattern [pattern=north east lines] (2,0) rectangle (3,2);
                    \pattern [pattern=north east lines] (1,2) rectangle (2,3);
                    \pattern [pattern=dots] (2,2) rectangle (3,3);
                    \pattern [pattern=vertical lines] (0,2) rectangle (1,3);
                \end{scope}
                \begin{scope}[canvas is yz plane at x=3]
                    \draw [step=1,help lines] (0,0) grid (3,3);
                \end{scope}
                \begin{scope}[canvas is xz plane at y=3]
                    \draw [step=1,help lines] (0,0) grid (3,3);
                \end{scope}
            \end{tikzpicture}
        }
        \caption{Operations to apply $\sigma^{-1} \pi^{-1} \sigma \pi$}
        \label{fig:[f,g]}
    \end{figure*}

    The $\pi \sigma \pi^{-1} \sigma^{-1}$ operation (if we may replace $\pi$, $\sigma$ with $\sigma^{-1}$, $\pi^{-1}$ respectively) provided by Lemma~\ref{thm:lem:op_[f,g]} is called of type-e1. It turns out to be effective as suggested by the following result.

    \begin{lemma}
        \label{thm:lem:[f,g]}
        Let $\tau$ be an even permutation in $A_{d}$. It takes $O(1)$ type-e1 operations to perform $\tau$ on a $1 \times 1 \times d$ sub-cuboid.
    \end{lemma}

    We will need a technical lemma before proving the above result.

    \begin{lemma}
        \label{thm:lem:decomp_commut_cycle}
        Any permutation $\tau \in S_{d}$ can be decomposed into three permutations $\tau = \tau_{3} \tau_{2} \tau_{1}$ such that all of $\tau_{1}, \tau_{2}, \tau_{3}$ have cycle decomposition only consisting of commutable $2$-cycles.
    \end{lemma}

    \begin{proof}
        Without loss of generality, assume that $\tau$ is a cyclic permutation $\tau = (a_{1} \ a_{2} \ \dots \ a_{s})$.
        Notice that
        \begin{equation*}
            (a_{1} \ a_{2} \ \dots \ a_{s}) = (a_{1} \ a_{2}) (a_{2} \ \dots\ a_{s}) = (a_{2} \ \dots \ a_{s}) (a_{s} \ a_{1}) .
        \end{equation*}
        Assume $s$ is odd, otherwise let $\tau_{1} = (a_{s} \ a_{1})$, $\tau' = (a_{2} \ \dots \ a_{s})$.
        According to the decomposition above, we have
        \begin{align*}
            (a_{1} \ a_{2} \ \dots \ a_{s}) &= (a_{1} \ a_{2}) (a_{2} \ \dots \ a_{s}) \\
            &= (a_{1} \ a_{2}) (a_{3} \ \dots \ a_{s}) (a_{s} \ a_{2}) \\
            &= (a_{1} \ a_{2}) (a_{s} \ a_{3} \ \dots \ a_{s-1}) (a_{s} \ a_{2}) .
        \end{align*}
        Perform the same decomposition on the smaller cycle $(a_{s} \ a_{3} \ \dots \ a_{s-1})$, we get
        \begin{equation*}
            (a_{1} \ a_{2}) (a_{s} \ a_{3}) (a_{s - 1} \ a_{4} \ \dots \ a_{s - 2}) (a_{s - 1} \ a_{3}) (a_{s} \ a_{2}) .
        \end{equation*}
        Notice that the $2$-cycles at the front (or the end) will never overlap.
        So we can recursively perform this decomposition until $s=1$ (since we assume that $s$ is odd).
        Now we get the desired
        \begin{align*}
            \tau_{2} &= (a_{\frac{s + 1}{2} + 1} \ a_{\frac{s + 1}{2}}) \dots (a_{s - 1} \ a_{3}) (a_{s} \ a_{2}) \\
            \tau_{3} &= (a_{1} \ a_{2}) (a_{s} \ a_{3}) (a_{s - 1} \ a_{4}) \dots (a_{\frac{s + 1}{2} + 2} \ a_{\frac{s + 1}{2}}) .
        \end{align*}
    \end{proof}

    Now we can prove the effectiveness of the $\pi \sigma \pi^{-1} \sigma^{-1}$ operation.

    \begin{proof}[Proof of Lemma~\ref{thm:lem:[f,g]}]
        According to Lemma~\ref{thm:lem:decomp_commut_cycle}, there is a decomposition $\tau = \tau_{3} \tau_{2} \tau_{1}$ such that all of $\tau_{1}, \tau_{2}, \tau_{3}$ have cycle decomposition only consisting of commutable $2$-cycles.
        Furthermore, since $\tau$ is an even permutation, there are even number of $2$-cycles in total in $\tau_{1}, \tau_{2}, \tau_{3}$. So we can assume they are all even permutations.

        By noticing that
        \begin{align*}
            (a_{3} \ a_{4}) (a_{1} \ a_{2}) &= (a_{1} \ a_{2} \ a_{4})^{-1} (a_{1} \ a_{2} \ a_{3})^{-1} \\
            & \phantom{{}={}} \quad (a_{1} \ a_{2} \ a_{4}) (a_{1} \ a_{2} \ a_{3}) ,
        \end{align*}
        we can implement two $2$-cycles simultaneously by the $\pi \sigma \pi^{-1} \sigma^{-1}$ operation.
        Furthermore, observe that as long as $\pi_{1}$ and $\sigma_{1}$ commute with $\pi_{2}$ and $\sigma_{2}$, we have
        \begin{align*}
            \pi_{1} \sigma_{1} \pi_{1}^{-1} \sigma_{1}^{-1} \cdot \pi_{2} \sigma_{2} \pi_{2}^{-1} \sigma_{2}^{-1} &= \pi_{1} \pi_{2} \cdot \sigma_{1} \sigma_{2} \\
            & \phantom{{}={}} \quad \cdot (\pi_{1} \pi_{2})^{-1} \cdot (\sigma_{1} \sigma_{2})^{-1}.
        \end{align*}
        Thus, we are able to implement $\tau_{1}$, $\tau_{2}$ or $\tau_{3}$ by one $\pi \sigma \pi^{-1} \sigma^{-1}$ operation.
        And this results in an implementation of $\tau$ by $3$ type-e1 operations and hence 12 basic operations.
    \end{proof}

    The operation provided by Lemma~\ref{thm:lem:[f,g]} is called of type-e2.

    \subsubsection{Manipulating a plane of the cube}

    In this subsubsection, we follow the same idea as the one in previous subsubsection and construct a sequence of operations that affect a plane of the cube instead of an edge.
    Combining the results in previous subsubsection, we can achieve arbitrary even permutation on a plane.

    \begin{lemma}
        \label{thm:lem:edges_on_plane}
        Let $\pi$, $\sigma_{i} (1 \leq i \leq d)$ be permutations in $S_{d}$.
        It takes $O(1)$ basic operations to do permutation $\sigma_{i}^{-1} \sigma_{\pi(i)}$ on row $i$ of a chosen plane simultaneously, $i = 1, \dots, d$.
    \end{lemma}

    \begin{proof}
        As shown by \figurename~\ref{fig:edges_on_plane}, we first apply $\pi$ to each column of the top face and then $\sigma_{i}$ to each row of the $i$-th plane from the front.
        Subsequently, we apply $\pi^{-1}$ to each column of the top face and then $\sigma_{i}^{-1}$ to each row of the $i$-th plane from the front.
        If a cubie is not on the top face, $\sigma_{i}$ and $\sigma_{i}^{-1}$ cancel out.
        If a cubie is on the the top face, it undergoes four transformations
        \begin{align*}
            (i, j) &\to (\pi(i), j) \\
            &\to (\pi(i), \sigma_{\pi(i)}(j)) \\
            &\to (i, \sigma_{\pi(i)}(j)) \\
            &\to (i, \sigma_{i}^{-1} \sigma_{\pi(i)}(j) ) .
        \end{align*}
    \end{proof}

    \begin{figure*}
        \centering
        \subfloat[\label{fig:edges_on_plane:a}]{
            \centering
            \begin{tikzpicture}[scale=0.95]
                \begin{scope}[canvas is xy plane at z=3]
                    \draw [step=1,help lines] (0,0) grid (3,3);
                \end{scope}
                \begin{scope}[canvas is yz plane at x=3]
                    \draw [step=1,help lines] (0,0) grid (3,3);
                    \draw[thick] (2.5, 1.5) ellipse (0.25 and 1.25);
                    \draw[->, thick] (2.25, 1.51) -- (2.25, 1.49) node[right] {$\pi$};
                \end{scope}
                \begin{scope}[canvas is xz plane at y=3]
                    \draw [step=1,help lines] (0,0) grid (3,3);
                    \pattern [pattern=dots] (0,2) rectangle (3,3);
                \end{scope}
            \end{tikzpicture}
        }
        \subfloat[\label{fig:edges_on_plane:b}]{
            \centering
            \begin{tikzpicture}[scale=0.95]
                \begin{scope}[canvas is xy plane at z=3]
                    \draw [step=1,help lines] (0,0) grid (3,3);
                \end{scope}
                \begin{scope}[canvas is yz plane at x=3]
                    \draw [step=1,help lines] (0,0) grid (3,3);
                \end{scope}
                \begin{scope}[canvas is xz plane at y=3]
                    \draw [step=1,help lines] (0,0) grid (3,3);
                    \pattern [pattern=dots] (0,1) rectangle (3,2);
                    \draw[thick] (1.5, 2.5) ellipse (1.25 and 0.25);
                    \draw[->, thick] (1.49, 2.25) -- (1.51, 2.25) node[left=1.5cm] {$\sigma_{i}$};
                    \draw[thick] (1.5, 1.5) ellipse (1.25 and 0.25);
                    \draw[->, thick] (1.49, 1.25) -- (1.51, 1.25);
                    \draw[thick] (1.5, 0.5) ellipse (1.25 and 0.25);
                    \draw[->, thick] (1.49, 0.25) -- (1.51, 0.25);
                \end{scope}
            \end{tikzpicture}
        }
        \subfloat[\label{fig:edges_on_plane:c}]{
            \centering
            \begin{tikzpicture}[scale=0.95]
                \begin{scope}[canvas is xy plane at z=3]
                    \draw [step=1,help lines] (0,0) grid (3,3);
                \end{scope}
                \begin{scope}[canvas is yz plane at x=3]
                    \draw [step=1,help lines] (0,0) grid (3,3);
                    \draw[thick] (2.5, 1.5) ellipse (0.25 and 1.25);
                    \draw[->, thick] (2.25, 1.40) -- (2.25, 1.41) node[right] {$\pi^{-1}$};
                \end{scope}
                \begin{scope}[canvas is xz plane at y=3]
                    \draw [step=1,help lines] (0,0) grid (3,3);
                    \pattern [pattern=dots] (0,1) rectangle (3,2);
                \end{scope}
            \end{tikzpicture}
        }
        \subfloat[\label{fig:edges_on_plane:d}]{
            \centering
            \begin{tikzpicture}[scale=0.95]
                \begin{scope}[canvas is xy plane at z=3]
                    \draw [step=1,help lines] (0,0) grid (3,3);
                \end{scope}
                \begin{scope}[canvas is yz plane at x=3]
                    \draw [step=1,help lines] (0,0) grid (3,3);
                \end{scope}
                \begin{scope}[canvas is xz plane at y=3]
                    \draw [step=1,help lines] (0,0) grid (3,3);
                    \pattern [pattern=dots] (0,2) rectangle (3,3);
                    \draw[thick] (1.5, 2.5) ellipse (1.25 and 0.25);
                    \draw[->, thick] (1.51, 2.25) -- (1.49, 2.25) node[left=1.5cm] {$\sigma_{i}^{-1}$};
                    \draw[thick] (1.5, 1.5) ellipse (1.25 and 0.25);
                    \draw[->, thick] (1.51, 1.25) -- (1.49, 1.25);
                    \draw[thick] (1.5, 0.5) ellipse (1.25 and 0.25);
                    \draw[->, thick] (1.51, 0.25) -- (1.49, 0.25);
                \end{scope}
            \end{tikzpicture}
        }
        \caption{Operations to apply $\sigma_{i}^{-1} \sigma_{\pi(i)}$}
        \label{fig:edges_on_plane}
    \end{figure*}

    The operation provided by Lemma~\ref{thm:lem:edges_on_plane} is called of type-p1.
    With type-p1 operation and type-e2 operation, we can construct type-1 and type-2 operation defined in Theorem~\ref{thm:thm01}.

    \begin{corollary}[Type-1 operation for a plane of the cube]
        \label{thm:cor:op1_plane}
        Let $\pi_{i} (1 \leq i \leq d)$ be even permutations in $S_{d}$.
        It takes $O(1)$ basic operations to do permutation $\pi_{i}$ on row $i$ of a chosen plane simultaneously, $i = 1, \dots, d$.
    \end{corollary}

    \begin{proof}
        Apply Lemma~\ref{thm:lem:edges_on_plane} first with $\pi = (1 \ \dots \ d)$, $\sigma_{1} = \mathrm{id}$, $\sigma_{i} = \sigma_{i - 1} \pi_{i - 1}$ ($1 < i \leq d$).
        Therefore, $\sigma_{i}^{-1} \sigma_{\pi(i)} = \sigma_{i}^{-1} \sigma_{i + 1} = \sigma_{i}^{-1} \sigma_{i} \pi_{i} = \pi_{i}$ for $1 \leq i < d$.
        Then, apply Lemma~\ref{thm:lem:[f,g]} to $d$-th row with $\tau = \pi_{d} \sigma_{d}$ as $\sigma_{d}^{-1} \sigma_{\pi(d)} = \sigma_{d}^{-1}$.
        It needs one type-p1 operation and one type-e2 operation. Hence, 16 basic operations suffice.
    \end{proof}

    \begin{corollary}[Type-2 operation for a plane of the cube]
        \label{thm:cor:op2_plane}
        Let $x_{1} \neq x_{2}$ be elements in $\{1, \dots, d\}$ and $y_{1}, \dots, y_{2k}$ be unique elements also in $\{1, \dots, d\}$.
        It takes $O(1)$ basic operations to swap $(x_{1}, y_{i})$ with $(x_{2}, y_{i})$ on the top face simultaneously, $i = 1, \dots, 2k$.
    \end{corollary}

    \begin{proof}
        Apply Lemma~\ref{thm:lem:edges_on_plane} with $\pi = (y_{1} \ y_{2}) \cdots (y_{2k - 1} \ y_{2k})$ and
        \begin{equation*}
            \sigma_{i} = \begin{cases}
                (x_{1} \ x_{2}), & i \in \{y_{2j - 1} : 1 \leq j \leq k\} \\
                \mathrm{id}, & \text{otherwise} .
            \end{cases}
        \end{equation*}
        If $y \notin \{y_{1}, \dots, y_{2k}\}$, $\sigma_{y} = \sigma_{\pi(y)} = \mathrm{id}$. If $y \in \{y_{1}, \dots, y_{2k}\}$, one of $\sigma_{y}$ and $\sigma_{\pi(y)}$ is $\mathrm{id}$ and the other one satisfies $\sigma = \sigma^{-1} = (x_{1} \ x_{2})$.
        Hence, 4 basic operations suffice.
    \end{proof}

    Simply apply Theorem~\ref{thm:thm01}, we can get the following corollary.

    \begin{corollary}
        \label{thm:cor:plane_even}
        It takes $O(1)$ basic operations to perform an arbitrary even permutation on a chosen plane.
    \end{corollary}

    \subsubsection{Main result}

    \begin{theorem}
        \label{thm:cor:cube_even}
        Any permutation in $A_{d^3}$ can be implemented by $O(d)$ 2-qudit reversible circuits.
    \end{theorem}
    
    \begin{proof}
        If a $d \times d \times d$ cube is viewed as a $d \times d^{2}$ plane, then type-1r operation is provided by using Corollary~\ref{thm:cor:plane_even} at most $d$ times.
        As for type-1c, using Corollary~\ref{thm:cor:plane_even} once can deal with $d$ columns.
        To handle $d^2$ columns, we need to use this corollary at most $d$ times.
        
        For type-2r operation for the $d \times d^{2}$ plane, we use Corollary~\ref{thm:cor:plane_even} to eliminate even swaps in a row.
        If any row requires odd swaps, we add one swap $(1 \ 2)$.
        As all these $(1 \ 2)$ swaps fall into a plane, we use Corollary~\ref{thm:cor:plane_even} another time.
        For type-2c operation, we divide $d^2$ columns into $d$ mega-columns, each of $d$ columns.
        Similarly, we use Corollary~\ref{thm:cor:plane_even} to eliminate even swaps in a mega-columns.
        If any mega-column requires odd swaps, we add swap $(1 \ 2)$ to the leftmost column of that mega-column.
        All these additional swaps fall into a plane, we use Corollary~\ref{thm:cor:plane_even} another time.
        
        Applying Theorem~\ref{thm:thm01}, we finish the proof.
    \end{proof}

    \subsection{Decomposition for the Case \texorpdfstring{$n > 3$}{n\textgreater 3}}

    We can view the $d \times \dots \times d$ hypercube as a $d \times d \times d^{n - 2}$ cube. Then we can adopt the technique developed in previous subsection. The validity lies in the following fact: the $d \times d^{n - 2}$ plane corresponds to a hyperplane in the original hypercube. As shown by \figurename~\ref{fig:4d_visual}, we can also choose the $d \times d$ plane as \figurename~\ref{fig:4d_visual:a} is equivalent to \figurename~\ref{fig:4d_visual:b}.

    \begin{figure}
        \centering
        \subfloat[\label{fig:4d_visual:a}]{
            \begin{tikzpicture}[scale=0.7]
                \begin{scope}[canvas is xy plane at z=2]
                    \draw[help lines, step=2] (0,0) grid (6,2);
                    \draw[pattern=dots] (0,0) rectangle (2,2);
                    \draw[pattern=vertical lines] (2,0) rectangle (4,2);
                    \draw[pattern=north east lines] (4,0) rectangle (6,2);
                \end{scope}
                \begin{scope}[canvas is yz plane at x=6]
                    \draw[help lines, step=2] (0,0) grid (2,2);
                    \draw[thick] (1,1) ellipse (0.75 and 0.75);
                    \draw[thick, ->] (0.99, 0.25) -- (1.01, 0.25) node[left] {$f$};
                \end{scope}
                \begin{scope}[canvas is xz plane at y=2]
                    \draw[help lines, step=2] (0,0) grid (6,2);
                \end{scope}
            \end{tikzpicture}
        }
        \\
        \subfloat[\label{fig:4d_visual:b}]{
            \begin{tikzpicture}[scale=0.7]
                \begin{scope}[canvas is xy plane at z=2]
                    \draw[help lines, step=2] (0,0) grid (2,6);
                    \draw[pattern=dots] (0,0) rectangle (2,2);
                    \draw[pattern=vertical lines] (0,2) rectangle (2,4);
                    \draw[pattern=north east lines] (0,4) rectangle (2,6);
                \end{scope}
                \begin{scope}[canvas is yz plane at x=2]
                    \draw[help lines, step=2] (0,0) grid (6,2);
                    \draw[thick] (1,1) ellipse (0.75 and 0.75);
                    \draw[thick, ->] (0.99, 0.25) -- (1.01, 0.25) node[left] {$f$};
                    \draw[thick] (3,1) ellipse (0.75 and 0.75);
                    \draw[thick, ->] (2.99, 0.25) -- (3.01, 0.25) node[left] {$f$};
                    \draw[thick] (5,1) ellipse (0.75 and 0.75);
                    \draw[thick, ->] (4.99, 0.25) -- (5.01, 0.25) node[left] {$f$};
                \end{scope}
                \begin{scope}[canvas is xz plane at y=6]
                    \draw[help lines, step=2] (0,0) grid (2,2);
                \end{scope}
            \end{tikzpicture}
        }
        \caption{Visualization of the state space of multi-qudit system}
        \label{fig:4d_visual}
    \end{figure}

    It is easy to verify that the following lemmas hold.
    \begin{itemize}
        \item Let $\pi$, $\sigma$ be permutation in $S_{d^{n - 2}}$.
        It takes $O(1)$ basic operations to do permutation $\pi \sigma \pi^{-1} \sigma^{-1}$ on a $1 \times 1 \times d^{n - 2}$ sub-cuboid.
        \item Let $\pi$, $\sigma$ be permutation in $S_{d}$.
        It takes $O(1)$ basic operations to do permutation $\pi \sigma \pi^{-1} \sigma^{-1}$ on a $1 \times 1 \times d$ sub-cuboid.
        \item Let $\pi$ be permutation in $S_{d}$, $\sigma_{i} (1 \leq i \leq d)$ be permutations in $S_{d^{n - 2}}$.
        It takes $O(1)$ basic operations to do permutation $\sigma_{i}^{-1} \sigma_{\pi(i)}$ on row $i$ of a chosen plane simultaneously, $i = 1, \dots, d$.
        \item Let $\pi$ be permutation in $S_{d^{n - 2}}$, $\sigma_{i} (1 \leq i \leq d^{n - 2})$ be permutations in $S_{d}$.
        It takes $O(1)$ basic operations to do permutation $\sigma_{i}^{-1} \sigma_{\pi(i)}$ on row $i$ of a chosen plane simultaneously, $i = 1, \dots, d^{n - 2}$.
        \item It takes $O(1)$ basic operations to perform an arbitrary even permutation on a chosen $d \times d^{n - 2}$ plane.
    \end{itemize}

    We can use the same tricks as in $n = 3$ case to provide the type-1 opeartion and type-2 operation for the hypercube.
    As a result, we can get the following corollary.
    \begin{theorem}
        Any permutation in $A_{d^{n}}$ can be implemented by $O(d)$ $(n - 1)$-qudit reversible circuits.
    \end{theorem}

    \subsection{Lower Bound}

    \begin{theorem}
        There exists $\pi \in A_{d^{n}}$ that needs to be implemented by $\Omega(d)$ $(n - 1)$-qudit reversible circuits.
    \end{theorem}

    \begin{proof}
        There are $d^{n - 1}!$ different reversible function from $\mathbb{Z}_{d^{n-1}}$ to $\mathbb{Z}_{d^{n-1}}$. And each sub-circuit has $\binom{n}{n - 1} = n$ choices of targeted qudits. If $N$ sub-circuits suffice to implement all permutations in $A_{d^{n}}$, then $\left(n \cdot d^{n - 1}! \right)^{N} \geq |A_{d^{n}}| = \frac{1}{2} d^{n}!$. Therefore,
        \begin{equation*}
            N \geq \dfrac{\ln (\frac{1}{2} \cdot d^{n}!)}{\ln (n \cdot d^{n - 1}!)} \geq \dfrac{n d^{n} \ln d - d^{n} - \ln 2}{(n - 1) d^{n - 1} \ln d + \ln n} = \Omega(d)
        \end{equation*}
        The ``$\geq$'' is due to the following inequality
        \begin{equation*}
            e^{-n} n^{n} < \sqrt{2 \pi} e^{-n} n^{n + \frac{1}{2}} < n! < n^{n} .
        \end{equation*}
    \end{proof}
    Therefore, our bound is asymptotically tight.

    \subsection{Proof of Lemma~\ref{thm:lem01weak}} \label{section:result01:lem01}

    We will prove stronger version of Lemma~\ref{thm:lem01weak} in this subsection.

    \begin{lemma}
        \label{thm:lem01strong}
        It takes $O(1)$ operations of type-2r' or type-2c' to swap $(x_{1}, y_{1})$ with $(x_{2}, y_{2})$ and $(\tilde{x}_{1}, \tilde{y}_{1})$ with $(\tilde{x}_{2}, \tilde{y}_{2})$ when $s, t \geq 3$.
        \begin{enumerate}
            \item[2r')] Select two rows $r_{1}, r_{2}$ and two columns $c_{1}, c_{2}$. Swap $(r_{i}, c_{1})$ with $(r_{i}, c_{2})$, $i = 1, 2$.
            \item[2c')] Select two columns $c_{1}, c_{2}$ and two rows $r_{1}, r_{2}$. Swap $(r_{1}, c_{i})$ with $(r_{2}, c_{i})$, $i = 1, 2$.
        \end{enumerate}
    \end{lemma}

    In this subsection, an operation refers to a type-2r' operation or a type-2c' operation.
    We will prove the following technical lemma first.

    \begin{lemma}
        Assume $(x_1, y_1), (x_1, y_2), (x_2, y_1), (x_2, \tilde{y}_2)$ are four unique points.
        It takes $3$ operations to swap $(x_1, y_1)$ with $(x_1, y_2)$ and $(x_2, y_1)$ with $(x_2, \tilde{y}_2)$.
    \end{lemma}

    \begin{proof}
        Denote $a = (x_1, y_1), a' = (x_1, y_2), b = (x_2, y_1), b' = (x_2, \tilde{y}_{2})$ in \figurename~\ref{fig:fig11:a}. We can take $\alpha = (x_2, y_2), \beta = (x_3, y_2), \gamma = (x_3, \tilde{y}_2)$. First, swap $\alpha$ with $b'$ and $\beta$ with $\gamma$. Then, swap $a$ with $a'$ and $b$ with $b'$. Finally, swap $b$ with $\alpha$ and $\beta$ with $\gamma$. Comparing \figurename~\ref{fig:fig11:a} and \figurename~\ref{fig:fig11:d}, $a$ is swapped with $a'$ and $b$ is swapped with $b'$ while $\alpha, \beta, \gamma$ remain unaffected.
        \begin{figure}
            \centering
            \subfloat[\label{fig:fig11:a}]{
                \centering
                \begin{tikzpicture}
                    \draw[step=1,help lines] (0,0) grid (3,3);
                    \node (A) at (0.5,0.5) {$a$};
                    \node (B) at (0.5,1.5) {$a'$};
                    \node (C) at (1.5,0.5) {$b$};
                    \node (D) at (1.5,2.5) {$b'$};
                    \node (U) at (1.5,1.5) {$\alpha$};
                    \node (V) at (2.5,1.5) {$\beta$};
                    \node (W) at (2.5,2.5) {$\gamma$};
                    \draw[<->] (D) -- (U);
                    \draw[<->] (V) -- (W);
                \end{tikzpicture}
            }
            \subfloat[\label{fig:fig11:b}]{
                \centering
                \begin{tikzpicture}
                    \draw[step=1,help lines] (0,0) grid (3,3);
                    \node (A) at (0.5,0.5) {$a$};
                    \node (B) at (0.5,1.5) {$a'$};
                    \node (C) at (1.5,0.5) {$b$};
                    \node (D) at (1.5,2.5) {$\alpha$};
                    \node (U) at (1.5,1.5) {$b'$};
                    \node (V) at (2.5,1.5) {$\gamma$};
                    \node (W) at (2.5,2.5) {$\beta$};
                    \draw[<->] (A) -- (B);
                    \draw[<->] (C) -- (U);
                \end{tikzpicture}
            }
            \\
            \subfloat[\label{fig:fig11:c}]{
                \centering
                \begin{tikzpicture}
                    \draw[step=1,help lines] (0,0) grid (3,3);
                    \node (A) at (0.5,0.5) {$a'$};
                    \node (B) at (0.5,1.5) {$a$};
                    \node (C) at (1.5,0.5) {$b'$};
                    \node (D) at (1.5,2.5) {$\alpha$};
                    \node (U) at (1.5,1.5) {$b$};
                    \node (V) at (2.5,1.5) {$\gamma$};
                    \node (W) at (2.5,2.5) {$\beta$};
                    \draw[<->] (D) -- (U);
                    \draw[<->] (V) -- (W);
                \end{tikzpicture}
            }
            \subfloat[\label{fig:fig11:d}]{
                \centering
                \begin{tikzpicture}
                    \draw[step=1,help lines] (0,0) grid (3,3);
                    \node (A) at (0.5,0.5) {$a'$};
                    \node (B) at (0.5,1.5) {$a$};
                    \node (C) at (1.5,0.5) {$b'$};
                    \node (D) at (1.5,2.5) {$b$};
                    \node (U) at (1.5,1.5) {$\alpha$};
                    \node (V) at (2.5,1.5) {$\beta$};
                    \node (W) at (2.5,2.5) {$\gamma$};
                \end{tikzpicture}
            }
            \caption{Operations to swap $(x_1, y_1)$ with $(x_1, y_2)$ and $(x_2, y_1)$ with $(x_2, \tilde{y}_2)$}
            \label{fig:fig11}
        \end{figure}
    \end{proof}

    Naturally, the following corollary holds.

    \begin{corollary}
        Assume $(x_1, y_1), (x_1, y_2), (x_2, \tilde{y}_1), (x_2, \tilde{y}_2)$ are four unique points.
        It takes $5$ operations to swap $(x_1, y_1)$ with $(x_1, y_2)$ and $(x_2, \tilde{y}_1)$ with $(x_2, \tilde{y}_2)$.
    \end{corollary}

    \begin{proof}
        Compared with \figurename~\ref{fig:fig11}, we need two extra steps:
        \begin{enumerate}
            \item[a2)] Bring $b$ to the same row of $a$.
            \item[c2)] Revert step a2 to bring $b'$ to the original row of $b$.
        \end{enumerate}
    \end{proof}

    Now, we can prove the lemma.

    \begin{proof}[Proof of Lemma~\ref{thm:lem01strong}]
        Denote $a = (x_1, y_1), a' = (x_2, y_2), b = (\tilde{x}_1, \tilde{y}_1), b' = (\tilde{x}_2, \tilde{y}_2)$.

        \begin{enumerate}
            \item $aa' \nparallel \text{ axis}$, $bb' \nparallel \text{ axis}$.

            \begin{enumerate}
                \item $a b a' b'$ is rectangle. 2 operations suffice as shown by \figurename~\ref{fig:fig12}. \label{case:case1}
                \begin{figure}
                    \centering
                    \subfloat[\label{fig:fig12:a}]{
                        \centering
                        \begin{tikzpicture}
                            \draw[step=1,help lines] (0,0) grid (3,3);
                            \node (A) at (0.5,0.5) {$a$};
                            \node (B) at (0.5,1.5) {$b$};
                            \node (C) at (1.5,0.5) {$b'$};
                            \node (D) at (1.5,1.5) {$a'$};
                            \draw [<->] (A) -- (B);
                            \draw [<->] (C) -- (D);
                        \end{tikzpicture}
                    }
                    \subfloat[\label{fig:fig12:b}]{
                        \centering
                        \begin{tikzpicture}
                            \draw[step=1,help lines] (0,0) grid (3,3);
                            \node (A) at (0.5,0.5) {$b$};
                            \node (B) at (0.5,1.5) {$a$};
                            \node (C) at (1.5,0.5) {$a'$};
                            \node (D) at (1.5,1.5) {$b'$};
                            \draw [<->] (A) -- (C);
                            \draw [<->] (B) -- (D);
                        \end{tikzpicture}
                    }
                    \\
                    \subfloat[\label{fig:fig12:c}]{
                        \centering
                        \begin{tikzpicture}
                            \draw[step=1,help lines] (0,0) grid (3,3);
                            \node (A) at (0.5,0.5) {$a'$};
                            \node (B) at (0.5,1.5) {$b'$};
                            \node (C) at (1.5,0.5) {$b$};
                            \node (D) at (1.5,1.5) {$a$};
                        \end{tikzpicture}
                    }
                    \caption{Case \ref{case:case1} of Lemma~\ref{thm:lem01strong}}
                    \label{fig:fig12}
                \end{figure}

                \item Otherwise. $5 + 5 + 5 = 15$ operations suffice as shown by \figurename~\ref{fig:fig13}.
                We can take $\alpha = (x_{1}, y_{2})$.
                If that coincides with $b$ or $b'$, change $\alpha$ to $(x_{2}, y_{1})$, which does not coincides with $b$ or $b'$ (otherwise $a b a' b'$ is rectangle).
                Similarly, we can find $\beta$ such that $b \beta \parallel \text{axis}$, $b' \beta \parallel \text{axis}$ and $\beta \notin \{ a, a' \}$.
                If $\alpha = \beta$, change both $\alpha$ and $\beta$ to their alternatives.
                \label{case:case2}
                \begin{figure}
                    \centering
                    \subfloat[\label{fig:fig13:a}]{
                        \centering
                        \begin{tikzpicture}
                            \draw[step=1,help lines] (0,0) grid (4,4);
                            \node (A) at (1.5,3.5) {$a$};
                            \node (B) at (3.5,2.5) {$a'$};
                            \node (C) at (0.5,1.5) {$b$};
                            \node (D) at (2.5,0.5) {$b'$};
                            \node (E) at (3.5,3.5) {$\alpha$};
                            \node (F) at (0.5,0.5) {$\beta$};
                            \draw[<->] (A) -- (E);
                            \draw[<->] (D) -- (F);
                        \end{tikzpicture}
                    }
                    \subfloat[\label{fig:fig13:b}]{
                        \centering
                        \begin{tikzpicture}
                            \draw[step=1,help lines] (0,0) grid (4,4);
                            \node (A) at (1.5,3.5) {$\alpha$};
                            \node (B) at (3.5,2.5) {$a'$};
                            \node (C) at (0.5,1.5) {$b$};
                            \node (D) at (2.5,0.5) {$\beta$};
                            \node (E) at (3.5,3.5) {$a$};
                            \node (F) at (0.5,0.5) {$b'$};
                            \draw[<->] (B) -- (E);
                            \draw[<->] (C) -- (F);
                        \end{tikzpicture}
                    }
                    \\
                    \subfloat[\label{fig:fig13:c}]{
                        \centering
                        \begin{tikzpicture}
                            \draw[step=1,help lines] (0,0) grid (4,4);
                            \node (A) at (1.5,3.5) {$\alpha$};
                            \node (B) at (3.5,2.5) {$a$};
                            \node (C) at (0.5,1.5) {$b'$};
                            \node (D) at (2.5,0.5) {$\beta$};
                            \node (E) at (3.5,3.5) {$a'$};
                            \node (F) at (0.5,0.5) {$b$};
                            \draw[<->] (A) -- (E);
                            \draw[<->] (D) -- (F);
                        \end{tikzpicture}
                    }
                    \subfloat[\label{fig:fig13:d}]{
                        \centering
                        \begin{tikzpicture}
                            \draw[step=1,help lines] (0,0) grid (4,4);
                            \node (A) at (1.5,3.5) {$a'$};
                            \node (B) at (3.5,2.5) {$a$};
                            \node (C) at (0.5,1.5) {$b'$};
                            \node (D) at (2.5,0.5) {$b$};
                            \node (E) at (3.5,3.5) {$\alpha$};
                            \node (F) at (0.5,0.5) {$\beta$};
                        \end{tikzpicture}
                    }
                    \caption{Case \ref{case:case2} of Lemma~\ref{thm:lem01strong}}
                    \label{fig:fig13}
                \end{figure}
            \end{enumerate}

            \item $aa' \parallel \text{ axis}$, $bb' \nparallel \text{ axis}$. $1 + 15 + 1 = 17$ operations suffice as shown by \figurename~\ref{fig:fig14}.
            The first step is to find $\alpha$, $\beta$, $\gamma$ such that they can form a rectangle with $a$ or $a'$.
            Without loss of generality, assume $a a'$ is horizontal and $a = (1, 1)$, $a' = (2, 1)$.
            Let $\gamma = (3, 1)$. (Note that $\gamma = b$ or $\gamma = b'$ may hold.)
            Find a $y_{0} > 1$ with largest $score(y) := | \{ (1, y), (2, y), (3, y) \} \setminus \{ b, b' \} |$.
            As $b b'$ is not parallel with axis, $score(y_{0}) \geq 2$.
            If $score(y_{0}) = 3$, let $\alpha = (2, y_{0})$ and $\beta = (3, y_{0})$.
            If $\gamma = b$ or $\gamma = b'$ holds, $score(y_{0})$ must be $3$ and $b b' \nparallel \text{axis}$ still holds if $\gamma$ is swapped with $\beta$.
            If $score(y_{0}) = 2$, there exists a valid $y_{0}$ such that $(3, y_{0}) \notin \{b, b'\}$.
            Let $\beta = (3, y_{0})$ and $\alpha$ be $(1, y_{0})$ or $(2, y_{0})$ accordingly.
            It converts to case~\ref{case:case2} through swapping $\alpha$ with $a'$ (or $a$) and $\beta$ with $\gamma$.

            \label{case:case3}
            \begin{figure}
                \centering
                \subfloat[\label{fig:fig14:a}]{
                    \centering
                    \begin{tikzpicture}
                        \draw[step=1,help lines] (0,0) grid (4,4);
                        \node (A) at (0.5,1.5) {$a$};
                        \node (B) at (2.5,1.5) {$a'$};
                        \node (C) at (1.5,0.5) {$b$};
                        \node (D) at (3.5,3.5) {$b'$};
                        \node (U) at (2.5,2.5) {$\alpha$};
                        \node (V) at (3.5,2.5) {$\beta$};
                        \node (W) at (3.5,1.5) {$\gamma$};
                        \draw[<->] (B) -- (U);
                        \draw[<->] (V) -- (W);
                    \end{tikzpicture}
                }
                \subfloat[\label{fig:fig14:b}]{
                    \centering
                    \begin{tikzpicture}
                        \draw[step=1,help lines] (0,0) grid (4,4);
                        \node (A) at (0.5,1.5) {$a$};
                        \node (B) at (2.5,1.5) {$\alpha$};
                        \node (C) at (1.5,0.5) {$b$};
                        \node (D) at (3.5,3.5) {$b'$};
                        \node (U) at (2.5,2.5) {$a'$};
                        \node (V) at (3.5,2.5) {$\gamma$};
                        \node (W) at (3.5,1.5) {$\beta$};
                        \draw[<->] (A) -- (U);
                        \draw[<->] (C) to[bend right] (D);
                    \end{tikzpicture}
                }
                \\
                \subfloat[\label{fig:fig14:c}]{
                    \centering
                    \begin{tikzpicture}
                        \draw[step=1,help lines] (0,0) grid (4,4);
                        \node (A) at (0.5,1.5) {$a'$};
                        \node (B) at (2.5,1.5) {$\alpha$};
                        \node (C) at (1.5,0.5) {$b'$};
                        \node (D) at (3.5,3.5) {$b$};
                        \node (U) at (2.5,2.5) {$a$};
                        \node (V) at (3.5,2.5) {$\gamma$};
                        \node (W) at (3.5,1.5) {$\beta$};
                        \draw[<->] (B) -- (U);
                        \draw[<->] (V) -- (W);
                    \end{tikzpicture}
                }
                \subfloat[\label{fig:fig14:d}]{
                    \centering
                    \begin{tikzpicture}
                        \draw[step=1,help lines] (0,0) grid (4,4);
                        \node (A) at (0.5,1.5) {$a'$};
                        \node (B) at (2.5,1.5) {$a$};
                        \node (C) at (1.5,0.5) {$b'$};
                        \node (D) at (3.5,3.5) {$b$};
                        \node (U) at (2.5,2.5) {$\alpha$};
                        \node (V) at (3.5,2.5) {$\beta$};
                        \node (W) at (3.5,1.5) {$\gamma$};
                    \end{tikzpicture}
                }
                \caption{Case \ref{case:case3} of Lemma~\ref{thm:lem01strong}}
                \label{fig:fig14}
            \end{figure}

            \item $aa' \parallel \text{ axis}$, $bb' \parallel \text{ axis}$.

            \begin{enumerate}
                \item $aa' \parallel bb'$. If $a, a', b, b'$ are not collinear, 5 operations suffice. For collinear case, we can use 1 operation to convert it to non-collinear case -- swap $a$ with $\tilde{a}$ and $a'$ with $\tilde{a}'$ where $a a' \tilde{a}' \tilde{a}$ is a rectangle.
                Finally, another 1 operation is used to convert it back.

                \item $aa' \perp bb'$. If $a, a', b$ are not collinear, $(3 + 3 + 3) \times 2 = 18$ operations suffice as shown by \figurename~\ref{fig:fig15}. For collinear case, we can use 1 operation to convert it to non-collinear case -- similarly swap $a$ with $\tilde{a}$ and $a'$ with $\tilde{a}'$ where $a a' \tilde{a}' \tilde{a}$ is a rectangle.
                Finally, another 1 operation is used to convert it back.
                \label{case:case4}

                \begin{figure}
                    \centering
                    \subfloat[\label{fig:fig15:a}]{
                        \centering
                        \begin{tikzpicture}
                            \draw[step=1,help lines] (0,0) grid (3,3);
                            \node (A) at (0.5,1.5) {$a$};
                            \node (B) at (2.5,1.5) {$a'$};
                            \node (C) at (1.5,0.5) {$b$};
                            \node (D) at (1.5,2.5) {$b'$};
                            \node (E) at (0.5,0.5) {$\alpha$};
                            \node (F) at (2.5,2.5) {$\beta$};
                            \draw[<->] (A) -- (E);
                            \draw[<->] (B) -- (F);
                        \end{tikzpicture}
                    }
                    \subfloat[\label{fig:fig15:b}]{
                        \centering
                        \begin{tikzpicture}
                            \draw[step=1,help lines] (0,0) grid (3,3);
                            \node (A) at (0.5,1.5) {$\alpha$};
                            \node (B) at (2.5,1.5) {$\beta$};
                            \node (C) at (1.5,0.5) {$b$};
                            \node (D) at (1.5,2.5) {$b'$};
                            \node (E) at (0.5,0.5) {$a$};
                            \node (F) at (2.5,2.5) {$a'$};
                            \draw[<->] (C) -- (E);
                            \draw[<->] (D) -- (F);
                        \end{tikzpicture}
                    }
                    \\
                    \subfloat[\label{fig:fig15:c}]{
                        \centering
                        \begin{tikzpicture}
                            \draw[step=1,help lines] (0,0) grid (3,3);
                            \node (A) at (0.5,1.5) {$\alpha$};
                            \node (B) at (2.5,1.5) {$\beta$};
                            \node (C) at (1.5,0.5) {$a$};
                            \node (D) at (1.5,2.5) {$a'$};
                            \node (E) at (0.5,0.5) {$b$};
                            \node (F) at (2.5,2.5) {$b'$};
                            \draw[<->] (A) -- (E);
                            \draw[<->] (B) -- (F);
                        \end{tikzpicture}
                    }
                    \subfloat[\label{fig:fig15:d}]{
                        \centering
                        \begin{tikzpicture}
                            \draw[step=1,help lines] (0,0) grid (3,3);
                            \node (A) at (0.5,1.5) {$b$};
                            \node (B) at (2.5,1.5) {$b'$};
                            \node (C) at (1.5,0.5) {$a$};
                            \node (D) at (1.5,2.5) {$a'$};
                            \node (E) at (0.5,0.5) {$\alpha$};
                            \node (F) at (2.5,2.5) {$\beta$};
                            \draw[<->] (A) -- (D);
                            \draw[<->] (B) -- (C);
                        \end{tikzpicture}
                    }
                    \\
                    \subfloat[\label{fig:fig15:e}]{
                        \centering
                        \begin{tikzpicture}
                            \draw[step=1,help lines] (0,0) grid (3,3);
                            \node (A) at (0.5,1.5) {$a'$};
                            \node (B) at (2.5,1.5) {$a$};
                            \node (C) at (1.5,0.5) {$b'$};
                            \node (D) at (1.5,2.5) {$b$};
                            \node (E) at (0.5,0.5) {$\alpha$};
                            \node (F) at (2.5,2.5) {$\beta$};
                        \end{tikzpicture}
                    }
                    \caption{Case \ref{case:case4} of Lemma~\ref{thm:lem01strong}. (a), (b) and (c) swap $a$ with $b$ and $a'$ with $b'$. (d) does essentially the same as the sequence of (a), (b) and (c).}
                    \label{fig:fig15}
                \end{figure}
            \end{enumerate}
        \end{enumerate}

        The most costly case is \ref{case:case4} and needs $20$ operations.
    \end{proof}

    \section{Decomposition into 2-qudit gates} \label{section:result02}

    In previous section, we decomposed the circuit in a top-down manner.
    In this section, we will construct circuit gadgets in a bottom-up fashion, harnessing the potential of ancilla qudits.

    In subsection A, we will first introduce the concept of $dC^{m}X$ gate.
    Addition of an extra control qudit for such gate can be easy.
    In subsection B, we will introduce the concept of $dX^{(m)}$ gate and prove that it can be decomposed as $dC^{m - 1}X$ gates.
    In subsection C, we will present an algorithm.
    Finally, we will show that our circuit size is almost optimal.

    \subsection{Construction of \texorpdfstring{$dC^{m}X$}{dC\textasciicircum mX} gate}

    We call $\ket{c_1 c_2 \cdots c_{m - 1} c_{m}} \mathhyphen U \ket{c_1 c_2 \cdots c_{m - 1} \tilde{c}_{m}} \mathhyphen U$ a $dC^{m}U$ gate.

    \begin{example}
        A simple example of $dC^{m}U$ gate is the $dC^{2}X$ gate shown by \figurename~\ref{fig:dC2X} where each line denotes a qudit and circle indicates a control qudit of a gate.
        \begin{figure}
            \centering
            \begin{tikzpicture}
                \begin{yquant}[register/separation=3mm]
                    qubit {$\ket{\reg_{\idx}}$} q[3];
                    controlbox {$0$} q[0];
                    controlbox {$0$} q[1];
                    box {$X_{01}$} q[2] | q[0], q[1];
                    controlbox {$0$} q[0];
                    controlbox {$1$} q[1];
                    box {$X_{01}$} q[2] | q[0], q[1];

                    align -;
                    text {=} (q);
                    align -;

                    controlbox {$0$} q[0];
                    box {$X_{01}$} q[1] | q[0];
                    controlbox {$0$} q[1];
                    box {$X_{01}$} q[2] | q[1];
                    controlbox {$0$} q[0];
                    box {$X_{01}$} q[1] | q[0];
                    controlbox {$0$} q[1];
                    box {$X_{01}$} q[2] | q[1];
                \end{yquant}
            \end{tikzpicture}
            \caption{Synthesis of $\ket{00} \mathhyphen X_{01} \ket{01} \mathhyphen X_{01}$}
            \label{fig:dC2X}
        \end{figure}
    \end{example}

    \begin{lemma}
        $dC^{m}U$ gate can be synthesized as 8 $dC^{m - 1}U$ gates if $U^2 = I$.
    \end{lemma}

    \begin{figure*}
        \centering
        \begin{tikzpicture}
            \begin{yquant}[register/separation=3mm]
                qubit {$\ket{\reg_{\idx}}$} q[4];
                controlbox {$0$} q[0];
                controlbox {$0$} q[1];
                controlbox {$0$} q[2];
                box {$U$} q[3] | q[0], q[1], q[2];
                controlbox {$0$} q[0];
                controlbox {$0$} q[1];
                controlbox {$1$} q[2];
                box {$U$} q[3] | q[0], q[1], q[2];

                align -;
                text {=} (q);
                align -;

                controlbox {$0$} q[0];
                controlbox {$0$} q[1];
                box {$X_{01}$} q[2] | q[0], q[1];
                controlbox {$0$} q[0];
                controlbox {$1$} q[1];
                box {$X_{01}$} q[2] | q[0], q[1];
                controlbox {$0$} q[0];
                controlbox {$0$} q[1];
                box {$X_{02}$} q[2] | q[0], q[1];
                controlbox {$0$} q[0];
                controlbox {$2$} q[1];
                box {$X_{02}$} q[2] | q[0], q[1];
                controlbox {$0$} q[0];
                controlbox {$0$} q[2];
                box {$U$} q[3] | q[0], q[2];
                controlbox {$0$} q[0];
                controlbox {$1$} q[2];
                box {$U$} q[3] | q[0], q[2];
                controlbox {$0$} q[0];
                controlbox {$0$} q[1];
                box {$X_{02}$} q[2] | q[0], q[1];
                controlbox {$0$} q[0];
                controlbox {$2$} q[1];
                box {$X_{02}$} q[2] | q[0], q[1];

                barrier (-);

                controlbox {$0$} q[0];
                controlbox {$0$} q[1];
                box {$X_{01}$} q[2] | q[0], q[1];
                controlbox {$0$} q[0];
                controlbox {$1$} q[1];
                box {$X_{01}$} q[2] | q[0], q[1];
                controlbox {$0$} q[0];
                controlbox {$0$} q[1];
                box {$X_{02}$} q[2] | q[0], q[1];
                controlbox {$0$} q[0];
                controlbox {$2$} q[1];
                box {$X_{02}$} q[2] | q[0], q[1];
                controlbox {$0$} q[0];
                controlbox {$0$} q[2];
                box {$U$} q[3] | q[0], q[2];
                controlbox {$0$} q[0];
                controlbox {$1$} q[2];
                box {$U$} q[3] | q[0], q[2];
                controlbox {$0$} q[0];
                controlbox {$0$} q[1];
                box {$X_{02}$} q[2] | q[0], q[1];
                controlbox {$0$} q[0];
                controlbox {$2$} q[1];
                box {$X_{02}$} q[2] | q[0], q[1];
            \end{yquant}
        \end{tikzpicture}
        \caption{Synthesis of $\ket{000} \mathhyphen U \ket{001} \mathhyphen U$ when $U^2 = I$}
        \label{fig:decomp_dCnX}
    \end{figure*}

    \begin{proof}
        As shown by \figurename~\ref{fig:decomp_dCnX}, we can neglect $q_{0}$ and focus on $q_{1}, q_{2}$ and $q_{3}$.
        If $q_{1}$ is not one of 0, 1 and 2, nothing will happen.

        If $q_{1} = 2$, as shown by \figurename~\ref{fig:decomp_dCnX:2}, two $X_{02}$ gates in the middle cancel each other out.
        Then two groups of controlled $U$ cancel each other out.
        Finally, the remaining two $X_{02}$ gates cancel each out.
        \begin{figure}
            \centering
            \begin{tikzpicture}
                \begin{yquant}[register/separation=3mm]
                    qubit {} q[4];

                    init {$\ket{q_0}$} q[0];
                    init {$\ket{q_1} = \ket{2}$} q[1];
                    init {$\ket{q_2}$} q[2];
                    init {$\ket{q_3}$} q[3];

                    controlbox {$0$} q[0];
                    box {$X_{02}$} q[2] | q[0];
                    controlbox {$0$} q[0];
                    controlbox {$0$} q[2];
                    box {$U$} q[3] | q[0], q[2];
                    controlbox {$0$} q[0];
                    controlbox {$1$} q[2];
                    box {$U$} q[3] | q[0], q[2];
                    controlbox {$0$} q[0];
                    box {$X_{02}$} q[2] | q[0];

                    barrier (-);

                    controlbox {$0$} q[0];
                    box {$X_{02}$} q[2] | q[0];
                    controlbox {$0$} q[0];
                    controlbox {$0$} q[2];
                    box {$U$} q[3] | q[0], q[2];
                    controlbox {$0$} q[0];
                    controlbox {$1$} q[2];
                    box {$U$} q[3] | q[0], q[2];
                    controlbox {$0$} q[0];
                    box {$X_{02}$} q[2] | q[0];
                \end{yquant}
            \end{tikzpicture}
            \caption{Equivalent quantum circuit of $\ket{000} \mathhyphen U \ket{001} \mathhyphen U$ when $U^2 = I$ and $\ket{q_1} = \ket{2}$}
            \label{fig:decomp_dCnX:2}
        \end{figure}

        If $q_{1} = 1$, as shown by \figurename~\ref{fig:decomp_dCnX:1}, nothing will happen unless $q_{2} = 0$ or $q_{2} = 1$ holds.
        Under the assumption $q_{2} = 0$ or $q_{2} = 1$, we can simplify $\ket{0} \mathhyphen U \ket{1} \mathhyphen U$ as $U$.
        Two $X_{01}$ gates get canceled out and two $U$ gates get canceled out.
        \begin{figure}
            \centering
            \begin{tikzpicture}
                \begin{yquant}[register/separation=3mm]
                    qubit {} q[4];

                    init {$\ket{q_0}$} q[0];
                    init {$\ket{q_1} = \ket{1}$} q[1];
                    init {$\ket{q_2}$} q[2];
                    init {$\ket{q_3}$} q[3];

                    controlbox {$0$} q[0];
                    box {$X_{01}$} q[2] | q[0];
                    controlbox {$0$} q[0];
                    controlbox {$0$} q[2];
                    box {$U$} q[3] | q[0], q[2];
                    controlbox {$0$} q[0];
                    controlbox {$1$} q[2];
                    box {$U$} q[3] | q[0], q[2];

                    barrier (-);

                    controlbox {$0$} q[0];
                    box {$X_{01}$} q[2] | q[0];
                    controlbox {$0$} q[0];
                    controlbox {$0$} q[2];
                    box {$U$} q[3] | q[0], q[2];
                    controlbox {$0$} q[0];
                    controlbox {$1$} q[2];
                    box {$U$} q[3] | q[0], q[2];
                \end{yquant}
            \end{tikzpicture}
            \caption{Equivalent quantum circuit of $\ket{000} \mathhyphen U \ket{001} \mathhyphen U$ when $U^2 = I$ and $\ket{q_1} = \ket{1}$}
            \label{fig:decomp_dCnX:1}
        \end{figure}

        The case of $q_{1} = 0$ is shown by \figurename~\ref{fig:decomp_dCnX:0}.
        \begin{itemize}
            \item $q_{2} = 2$. It gets mapped to $0$ by the first $X_{02}$ gate.
            Then an $U$ gate is applied to $q_{3}$.
            The second $X_{02}$ gate maps $q_{2}$ to $0$, which will be reverted by the third $X_{02}$ gate.
            Another $U$ gate applied to $q_{3}$ cancels the first one.
            $q_{2}$ is finally mapped to $2$ by the fourth $X_{02}$ gate.
            \item $q_{2} = 1$. It gets mapped to $2$ by the first $X_{01}$ gate and the first $X_{02}$ gate, which will be reverted by the second $X_{02}$ gate and the second $X_{01}$ gate.
            Then an $U$ gate is applied to $q_{3}$.
            \item $q_{2} = 0$. It gets mapped to $1$ by the first $X_{01}$ gate.
            Then an $U$ gate is applied to $q_{3}$.
            The second $X_{01}$ maps $q_{2}$ to $0$ and the third $X_{02}$ maps it to $2$.
            It is finally mapped to $0$ by the fourth $X_{02}$ gate.
        \end{itemize}
        \begin{figure}
            \centering
            \begin{tikzpicture}[scale=0.8]
                \begin{yquant}[register/separation=3mm]
                    qubit {} q[4];

                    init {$\ket{q_0}$} q[0];
                    init {$\ket{q_1} = \ket{0}$} q[1];
                    init {$\ket{q_2}$} q[2];
                    init {$\ket{q_3}$} q[3];

                    controlbox {$0$} q[0];
                    box {$X_{01}$} q[2] | q[0];
                    controlbox {$0$} q[0];
                    box {$X_{02}$} q[2] | q[0];
                    controlbox {$0$} q[0];
                    controlbox {$0$} q[2];
                    box {$U$} q[3] | q[0], q[2];
                    controlbox {$0$} q[0];
                    controlbox {$1$} q[2];
                    box {$U$} q[3] | q[0], q[2];
                    controlbox {$0$} q[0];
                    box {$X_{02}$} q[2] | q[0];

                    barrier (-);

                    controlbox {$0$} q[0];
                    box {$X_{01}$} q[2] | q[0];
                    controlbox {$0$} q[0];
                    box {$X_{02}$} q[2] | q[0];
                    controlbox {$0$} q[0];
                    controlbox {$0$} q[2];
                    box {$U$} q[3] | q[0], q[2];
                    controlbox {$0$} q[0];
                    controlbox {$1$} q[2];
                    box {$U$} q[3] | q[0], q[2];
                    controlbox {$0$} q[0];
                    box {$X_{02}$} q[2] | q[0];
                \end{yquant}
            \end{tikzpicture}
            \caption{Equivalent quantum circuit of $\ket{000} \mathhyphen U \ket{001} \mathhyphen U$ when $U^2 = I$ and $\ket{q_1} = \ket{0}$}
            \label{fig:decomp_dCnX:0}
        \end{figure}
    \end{proof}

    \subsection{Construction of \texorpdfstring{$dX^{(m)}$}{dX\textasciicircum (m)} gate}

    We call a pair of $X^{(m)}$ gates $X_{\mathbf{x}, \tilde{\mathbf{x}}} X_{\mathbf{y}, \tilde{\mathbf{y}}}$ as $dX^{(m)}$ gate.
    For $m = 2$, Lemma~\ref{thm:lem01strong} tells us that $dX^{(2)}$ can be constructed using $O(1)$ $dCX$ gates.
    If we can prove that $dX^{(m + 1)}$ can be constructed using $CdX^{(m)}$, then $dX^{(m + 1)}$ can be constructed using $C^{m - 1}dX^{(2)}$, thus $dC^{m}X$.

    \begin{lemma} \label{thm:lem:dxm}
        $dX^{(m + 1)}$ can be synthesized as 8 $CdX^{(m)}$ gates.
    \end{lemma}

    \begin{proof}
        Assume $|\mathbf{x}| = |\tilde{\mathbf{x}}| = |\mathbf{y}| = |\tilde{\mathbf{y}}| = m + 1$.
        Our goal is to implement $X_{\mathbf{x}, \tilde{\mathbf{x}}} X_{\mathbf{y}, \tilde{\mathbf{y}}}$.
        Our construction has three steps.
        The first step is to prove that $X_{\mathbf{x}, (\tilde{\mathbf{x}}_{1:m}, x_{m + 1})} X_{\mathbf{y}, (y_{1}, \tilde{\mathbf{y}}_{2:m+1})}$ can be synthesized as 2 $CdX^{(m)}$ gates.
        Note that $\mathbf{x} = (\mathbf{x}_{1:m}, x_{m+1})$ has the same $(m+1)$-th coordinate with $(\tilde{\mathbf{x}}_{1:m}, x_{m+1})$ and $\mathbf{y}$ has the same 1st coordinate with $(y_{1}, \tilde{\mathbf{y}}_{2:m+1})$.
        The second step is to remove one of the two constraints and the third step is to remove the other one constraint.
        Both steps need two gates constructed in their previous step.

        For the first step, we can interpret $X_{\mathbf{x}, (\tilde{\mathbf{x}}_{1:m}, x_{m + 1})}$ and $X_{\mathbf{y}, (y_{1}, \tilde{\mathbf{y}}_{2:m+1})}$ as $\ket{x_{m + 1}} \mathhyphen X_{\mathbf{x}_{1:m}, \tilde{\mathbf{x}}_{1:m}}$ and $\ket{y_{1}} \mathhyphen X_{\mathbf{y}_{2:m+1}, \tilde{\mathbf{y}}_{2:m+1}}$ respectively.
        We can insert two
        \begin{align*}
            & \phantom{{}={}} X_{(y_{1}, \mathbf{z}_{2:m}, x_{m + 1}), (y_{1}, \tilde{\mathbf{z}}_{2:m}, x_{m + 1})} \\
            &= \ket{x_{m + 1}} \mathhyphen X_{(y_{1}, \mathbf{z}_{2:m}), (y_{1}, \tilde{\mathbf{z}}_{2:m})} \\
            &= \ket{y_{1}} \mathhyphen X_{(\mathbf{z}_{2:m}, x_{m + 1}), (\tilde{\mathbf{z}}_{2:m}, x_{m + 1})}
        \end{align*}
        gates in between as $X_{(y_{1}, \mathbf{z}_{2:m}, x_{m + 1}), (y_{1}, \tilde{\mathbf{z}}_{2:m}, x_{m + 1})}$ cancels with itself. Thus,
        \begin{align*}
            & \phantom{{}={}} X_{\mathbf{x}, (\tilde{\mathbf{x}}_{1:m}, x_{m + 1})} X_{\mathbf{y}, (y_{1}, \tilde{\mathbf{y}}_{2:m+1})} \\
            &= \ket{x_{m+1}} \mathhyphen \left( X_{\mathbf{x}_{1:m}, \tilde{\mathbf{x}}_{1:m}} X_{(y_{1}, \mathbf{z}_{2:m}), (y_{1}, \tilde{\mathbf{z}}_{2:m})} \right) \\
            & \phantom{{}={}} \cdot \ket{y_{1}} \mathhyphen \left( X_{(\mathbf{z}_{2:m}, x_{m + 1}), (\tilde{\mathbf{z}}_{2:m}, x_{m + 1})} X_{\mathbf{y}_{2:m+1}, \tilde{\mathbf{y}}_{2:m+1}} \right)
        \end{align*}
        can be synthesized as 2 $CdX^{(m)}$ gates.

        \begin{figure}
            \centering
            \begin{tikzpicture}
                \begin{yquant}[register/separation=3mm]
                    qubit {} q[2];

                    box {$X_{\mathbf{y}, \tilde{\mathbf{y}}}$} (q);

                    text {=} (-);

                    slash q[0];
                    controlbox {$\mathbf{y}_{1:m}$} q[0];
                    box {$X_{y_{m + 1}, \tilde{y}_{m + 1}}$} q[1] | q[0];
                    controlbox {$\tilde{y}_{m + 1}$} q[1];
                    box {$X_{\mathbf{y}_{1:m}, \tilde{\mathbf{y}}_{1:m}}$} q[0] | q[1];
                    controlbox {$\mathbf{y}_{1:m}$} q[0];
                    box {$X_{y_{m + 1}, \tilde{y}_{m + 1}}$} q[1] | q[0];
                \end{yquant}
            \end{tikzpicture}
            \caption{Decomposition of $X_{\mathbf{y}, \tilde{\mathbf{y}}}$}
            \label{fig:decomp_CnX}
        \end{figure}

        For the second step, we prove that $X_{\mathbf{x}, (\tilde{\mathbf{x}}_{1:m}, x_{m + 1})} X_{\mathbf{y}, \tilde{\mathbf{y}}}$ can be synthesized as 2 gates of form $X_{\mathbf{x}, (\tilde{\mathbf{x}}_{1:m}, x_{m + 1})} X_{\mathbf{y}, (y_{1}, \tilde{\mathbf{y}}_{2:m+1})}$.
        The essence lies in \figurename~\ref{fig:decomp_CnX}, which is inspired by the decomposition of SWAP gate.
        It can be written as the following equation
        \begin{equation*}
            X_{\mathbf{y}, \tilde{\mathbf{y}}} = X_{\mathbf{y}, (\tilde{y}_{1}, \mathbf{y}_{2:m+1})} X_{(\tilde{y}_{1}, \mathbf{y}_{2:m+1}), \tilde{\mathbf{y}}} X_{\mathbf{y}, (\tilde{y}_{1}, \mathbf{y}_{2:m+1})} .
        \end{equation*}
        Adding $X_{\mathbf{x}, (\tilde{\mathbf{x}}_{1:m}, x_{m + 1})}$ gate to the left of both sides of the equation, we obtain that
        \begin{align*}
            & \phantom{{}={}} X_{\mathbf{x}, (\tilde{\mathbf{x}}_{1:m}, x_{m + 1})} X_{\mathbf{y}, \tilde{\mathbf{y}}} \\
            &= X_{\mathbf{x}, (\tilde{\mathbf{x}}_{1:m}, x_{m + 1})} X_{\mathbf{y}, (\tilde{y}_{1}, \mathbf{y}_{2:m+1})} \\
            & \phantom{{}={}} \cdot X_{\tilde{\mathbf{y}}, (\tilde{y}_{1}, \mathbf{y}_{2:m+1})} X_{\mathbf{y}, (\tilde{y}_{1}, \mathbf{y}_{2:m+1})}
        \end{align*}
        can be synthesized as two gates constructed in the first step.

        Finally, for $X_{\mathbf{x}, \tilde{\mathbf{x}}} X_{\mathbf{y}, \tilde{\mathbf{y}}}$, insert two $X_{\mathbf{z}, (\tilde{\mathbf{z}}_{1:m}, z_{m + 1})}$ gates in between.
        \begin{equation*}
            X_{\mathbf{x}, \tilde{\mathbf{x}}} X_{\mathbf{y}, \tilde{\mathbf{y}}} = X_{\mathbf{x}, \tilde{\mathbf{x}}} X_{\mathbf{z}, (\tilde{\mathbf{z}}_{1:m}, z_{m + 1})} \cdot X_{\mathbf{z}, (\tilde{\mathbf{z}}_{1:m}, z_{m + 1})} X_{\mathbf{y}, \tilde{\mathbf{y}}}
        \end{equation*}
        Thus, we construct $dX^{(m + 1)}$ using $2 \times 2 \times 2 = 8$ $CdX^{(m)}$ gates.
    \end{proof}

    As long as $m$ is a constant, $dX^{(m)}$ can be synthesized as $O(1)$ 2-qudit gates. Specifically, we are interested in the case $m = 4$, which serves as an important block in our algorithm presented in next subsection.

    \subsection{Synthesis of \texorpdfstring{$d$}{d}-ary reversible function}

    Our algorithm makes use of the following result suggested by Zi~et~al.~\cite{ziOptimalSynthesisMulticontrolled2023}.

    \begin{lemma}[\cite{ziOptimalSynthesisMulticontrolled2023}] \label{thm:lem:cnx}
        $C^{n - 1}X$ can be synthesized as $O(n d^3)$ 2-qudit gates using 1 dirty ancilla.
    \end{lemma}

    The main idea is to transform $O(d^3)$ elements into a sub-cube and then apply the lemma. On average, we spend $O(n)$ gates to move each element.

    \begin{theorem}
        Utilizing one clean ancilla qudit, any $n$-dit reversible function $f$ can be synthesized as $O(n d^n)$ 2-qudit gates.
    \end{theorem}

    \begin{proof}
        As shown by Algorithm~\ref{alg:alg01}.
        In a round, if a quantum state is not affected by $X_{(\tilde{p}_{j}, 1), (\widetilde{f(p_{j})}, 1)} X_{(\tilde{p}_{j'}, 1), (\widetilde{f(p_{j'})}, 1)}$, the effect of other quantum gates will be canceled.
        If $X_{(\tilde{p}_{j}, 1), (\widetilde{f(p_{j})}, 1)} X_{(\tilde{p}_{j'}, 1), (\widetilde{f(p_{j'})}, 1)}$ takes effect, we claim that $p_{j}$ will be swapped with $f(p_{j})$, which reduces $|\{ x : f(x) \neq x \}|$ by at least one as
            \begin{align*}
                & \phantom{{}={}} \left( p_{j} \ f(p_{j}) \ f(f(p_{j})) \ \dots \ f^{-1}(p_{j}) \right) \left( p_{j} \ f(p_{j}) \right) \\
                &= \left( p_{j} \ f(f(p_{j})) \ \dots \ f^{-1}(p_{j}) \right) .
            \end{align*}

        In each round, we construct $\mathcal{P}$ as follows.
        Initially, $\mathcal{P} = \varnothing$.
        Find a $p_{j} \notin \mathcal{P}$ such that $f(p_{j}) \neq p_{j}$ and $f(p_{j}) \notin \mathcal{P}$.
        Add $p_{j}$ and $f(p_{j})$ to $\mathcal{P}$.
        Repeat the process until there is no valid $p_{j}$ or $|\mathcal{P}| \geq d^{3} - 1$.
        If $|\mathcal{P}| / 2$ is odd, remove the last pair of $p_{j}$ and $f(p_{j})$.

        Then, we make the first $(n - 3)$ coordinates of points in $\mathcal{P}$ to be $0$ with $(n - 3)$ transformations.
        Each of these transformation acts on $4$ qudits.
        We find two points $q_{1}, q_{2}$ in $\mathcal{P}$ with first coordinate non-zero.
        If there is only one valid point, let $q_{2}$ be arbitrary point with first coordinate non-zero.
        We find two points $\tilde{q}_{1}, \tilde{q}_{2}$ with zero first coordinate that are not occupied by $\mathcal{P}$.
        The existence of $\tilde{q}_{1}$ and $\tilde{q}_{2}$ is guaranteed by $|\mathcal{P}| \leq d^{3}$.
        Swap $q_{1}$ with $\tilde{q}_{1}$ and $q_{2}$ with $\tilde{q}_{2}$.
        After $(n - 3)$ transformation, $p_{j}$ is transformed into $\tilde{p}_{j}$ with first $(n - 3)$ coordinates zero.
        Denote transformed $f(p_{j})$ as $\widetilde{f(p_{j})}$.

        Next, a $C^{n-3} X_{01}$ gate is applied. Now, the ancilla qudit is not $\ket{0}$ if and only if the first $(n - 3)$ coordinates are zero.
        We pair $\tilde{p}_{j}$ with $\tilde{p}_{j'}$ and apply $X_{(\tilde{p}_{j}, 1), (\widetilde{f(p_{j})}, 1)} X_{(\tilde{p}_{j'}, 1), (\widetilde{f(p_{j'})}, 1)}$. (Here, we abuse the notation. The first $(n - 3)$ coordinates of $\tilde{p}_{j}$ or $\widetilde{f_{p_{j}}}$ are ignored.)
        We actually swap $\tilde{p}_{j}$ with $\widetilde{f_{p_{j}}}$ and $\tilde{p}_{j'}$ with $\widetilde{f_{p_{j'}}}$ as $(\tilde{p}_{j,n-2}, \tilde{p}_{j,n-1}, \tilde{p}_{j,n}, 1)$ represents $(0, \dots, 0, \tilde{p}_{j,n-2}, \tilde{p}_{j,n-1}, \tilde{p}_{j,n}) = \tilde{p}_{j}$.
        We apply a $C^{n-3} X_{01}$ gate to restore the ancilla qudit.
        Inverting the $(n - 3)$ transformations, we obtain $p_{j} \mapsto \tilde{p}_{j} \mapsto \widetilde{f(p_{j})} \mapsto f(p_{j})$.
        Finally, $|\{ x : f(x) \neq x \}| \leq 3$. This can be done by at most 2 2-cycles and hence at most 2 $X^{(n)}$ gates.

        $O(n d^3)$ 2-qudit gates are needed in each round and there are at most $O(d^{n - 3})$ rounds.
        The final $X^{(n)}$ gate(s) can be implemented using the decomposition in \figurename~\ref{fig:decomp_CnX} when there exists $i \neq j$ such that $y_{i} \neq \tilde{y}_{i}$ and $y_{j} \neq \tilde{y}_{j}$.
        Two $C^{n - 1}X$ gates and $(n - 1)$ $CX$ gates need $O(n d^3)$ 2-qudit gates.
        Thus, $O(n d^n)$ gates are enough.
    \end{proof}

    \begin{algorithm}
        \caption{Synthesis of $n$-dit reversible function}
        \begin{algorithmic}[1]
            \Require{Reversible function $f : \mathbb{Z}_{d^{n}} \to \mathbb{Z}_{d^{n}}$}
            \Ensure{Qudit circuit for $f$}
            \While{$|\{ x : f(x) \neq x \}| \geq 4$}
            \State Select $O(d^3)$ points $p_j$ and corresponding $f(p_j)$. Denote the point set as $\mathcal{P}$.
            \For{$i \gets 1 \dots n - 3$} \label{alg:alg01:transform1:begin}
            \State Apply $O(d^3)$ $dX^{(4)}$ gates to qudit $i, i+1, i+2, i+3$, transforming the $i$-th coordinate of $\mathcal{P}$ to 0.
            \EndFor
            \State Denote transformed $p_{j}$ as $\tilde{p}_{j}$, transformed $f(p_{j})$ as $\widetilde{f(p_{j})}$
            \State Apply $C^{n-3} X_{01}$ with control $1, 2, \cdots, n - 3$ and target $n + 1$.
            \State Apply $O(d^3)$ $X_{(\tilde{p}_{j}, 1), (\widetilde{f(p_{j})}, 1)} X_{(\tilde{p}_{j'}, 1), (\widetilde{f(p_{j'})}, 1)}$ gates to qudit $n - 2, n - 1, n, n + 1$.
            \State Invert the $C^{n-3} X_{01}$ and ($n - 3$) transformations.
            \State $f \gets f \circ \left( \prod_{j} (p_{j} \ f(p_{j})) \right)$
            \EndWhile
            \While{$|\{ x : f(x) \neq x \}| \neq 0$}
            \State Apply $X^{(n)}$ gate.
            \EndWhile
        \end{algorithmic}
        \label{alg:alg01}
    \end{algorithm}

    \subsection{Lower bound}

    \begin{theorem}\label{thm:lowerbound}
        Let $L$ be any universal gate set consisting of 1-qudit and 2-qudit gates that can implement every permutation $\pi \in S_{d^n}$. If $|L| = \mathsf{poly}(n, d)$, then there exists $\pi_{0} \in S_{d^n}$ that requires $\Omega \left( n d^n \frac{\log d}{\log d + \log n} \right)$ gates in $L$ even with $\mathsf{poly}(n)$ ancilla.
    \end{theorem}

    \begin{proof}
        The proof relies on a counting argument. Suppose $|L| = p_1(n, d)$ and $p_2(n)$ ancilla are available, where $p_1(\cdot,\cdot)$ and $p_2(\cdot)$ are polynomials functions. A 1-qudit gate $G \in L$ can be put on $n + p_2(n)$ different qudits and make $n + p_2(n)$ different circuits of only 1 gate. Similarly, 2-qudit gate $G' \in L$ can choose $2$ target qubits and make at most $(n + p_2(n))(n + p_2(n) - 1)$ different circuits of only 1 gate.
        There are at most $(n + p_2(n))(n + p_2(n) - 1) |L| + 1$ different circuits of at most 1 gate.
        Therefore, there are at most $\left( (n + p_2(n))(n + p_2(n) - 1) |L| + 1 \right)^{N}$ different circuits of at most $N$ gates.
        Let $p_3(n, d) = (n + p_2(n))(n + p_2(n) - 1) p_1(n, d) + 1$.
        $p_3$ is still a polynimial of $n$ and $d$.
        If $N$ gates suffice to implement any $\pi \in S_{d^n}$, then $\left( p_3(n, d) \right)^{N} \geq |S_{d^n}| = (d^n)!$.
        This is equivalent to $N \geq \frac{\log((d^n)!)}{\log(p_3(n, d))} = \frac{\Omega( n d^n \log d )}{\log(p_3(n, d))} = \Omega \left( n d^n \frac{\log d}{\log d + \log n} \right)$.
        As a result, there exists $\pi_{0} \in S_{d^n}$ that requires $\Omega \left( n d^n \frac{\log d}{\log d + \log n} \right)$ gates in $L$.
    \end{proof}

    A direct corollary of Theorem~\ref{thm:lowerbound} is that our Algorithm~\ref{alg:alg01} is asymptotically optimal when $d = \Omega(n)$.

    \section{Conclusion} \label{section:concl}

    In this paper, we investigate the decomposition of reversible functions.
    We proved that an arbitrary $n$-dit even reversible function
    can be implemented with $\Theta(d)$ reversible subcircuits operating on $(n - 1)$-dits.
    This is achieved by developing a canonical decomposition of alternating group $A_{d^{n}}$ with $\Theta(d)$ permutations where each permutation has the form $\pi_{-k}\otimes \mathsf{id}_k$, where $\pi_{-k}\in S_{d^{n-1}}$ is a permutation acting on all dimensions except the $k$-th one.
    Additionally, we established that $n$-dit reversible function oracle can be implemented with $O \left( n d^{n} \right)$ gates using only $1$ ancilla.

    The following two findings are noteworthy and may be of independent interest.
    First, \cite{sunGeneralizedShuffleexchangeProblem2022} proved that permutation in $S_{s \times t}$ can be decomposed into a constant number of permutations, where each of these permutations acts either on the rows or on the columns of the elements. In this paper, we strengthen that result by relaxing the allowed permutations to even ones (Theorem~\ref{thm:thm01}), which reveal another combinatorial nature of two-dimensional permutations.
    Second, we proved that the $n$-qudit operation $X_{\mathbf{a}, \mathbf{b}}$, whose effect is to exchange the amplitude of computational basis $\ket{\mathbf{a}}$ and $\ket{\mathbf{b}}$, can be implemented using amortized $O(n)$ gates and $1$ ancilla, which may be helpful in the circuit design of reversible oracles.

    Some natural questions arise from our result:
    \begin{itemize}
        \item Can $n$-dit even reversible function be implemented using only even reversible subcircuits on $(n - 1)$ qudits?
        \item Can $n$-dit reverisible function be implemented using reversible subcircuits on $(n - 1)$ qudits when $d$ is odd?
        \item To synthesize $n$-dit reversible function with 2-qudit gates, our result is asymptotically optimal when $d = \Omega(n)$ and \cite{wuAsymptoticallyOptimalSynthesis2024} is asymptotically optimal when $d = o(\log n)$. What is the case when $d = o(n)$ and $d = \Omega(\log n)$?
        \item Develop more applications of the combinatorial decomposition results of the symmetric group.
    \end{itemize}


    \appendices

    \section{A toy example for Algorithm~\ref{alg:alg01}}

    \begin{figure}
        \centering
        \begin{tikzpicture}[scale=0.8]
            \begin{yquant}[register/separation=3mm]
                qubit {$\ket{\reg_{\idx}}$} q[4];
                qubit {$\ket{0}$} a;

                box {$X_{1007, 0001}$ \\ $X_{1042, 0002}$} (q);
                controlbox {$0$} q[0];
                box {$X_{0, 1}$} a | q[0];
                controlbox {1} a;
                box {$X_{007, 001}$ \\ $X_{042, 002}$} (q[1-3]) | a;
                controlbox {$0$} q[0];
                box {$X_{0, 1}$} a | q[0];
                box {$X_{1007, 0001}$ \\ $X_{1042, 0002}$} (q);

                barrier (-);

                controlbox {$0$} q[1];
                controlbox {$4$} q[2];
                controlbox {$2$} q[3];
                box {$X_{0, 2}$} q[0] | q[1-3];
            \end{yquant}
        \end{tikzpicture}
        \caption{Quantum circuit for the toy example}
        \label{fig:alg1-example}
    \end{figure}

    In this section, we construct a toy example to better explaim Algorithm~\ref{alg:alg01}. Take $n = 4$, $d = 10$ and $f = (0007 ~ 1007) (0042 ~ 1042 ~ 2042)$.
    \begin{equation*}
        f(x) = \begin{cases}
            1007, & x = 0007, \\
            0007, & x = 1007, \\
            1042, & x = 0042, \\
            2042, & x = 1042, \\
            0042, & x = 2042, \\
            x, & \text{otherwise} .
        \end{cases} 
    \end{equation*}
    We first pick $p_{1} = 0007$, $f(p_{1}) = 1007$, $p_{2} = 0042$, $f(p_{2}) = 1042$ as $\mathcal{P}$.
    As shown by \figurename~\ref{fig:alg1-example}, we apply $X_{1007, 0001} X_{1042, 0002}$.
    $\mathcal{P}$ is transformed into $\{ 0007, 0001, 0042, 0002 \}$ with the first coordinate to be $0$.
    The next three gates swap $0007$ with $0001$ and $0042$ with $0002$ using the help of ancilla. ($X_{007, 001} X_{042, 002}$ is enabled only if the first $n - 3 = 1$ coordinate is $0$.)
    Then we apply the reverse of $X_{1007, 0001} X_{1042, 0002}$, which is itself.
    Now we update $f$ to $(0042 ~ 2042)$.
    As $f$ has only 2 fixed points, we apply $X_{0042, 2042}$ directly.

    $X_{1007, 0001} X_{1042, 0002}$ and $\ket{1} \mathhyphen (X_{007, 001} X_{042, 002})$ can be implemented with $O(1)$ 2-qudit gates according to Lemma~\ref{thm:lem:dxm}. And $\ket{042} \mathhyphen X_{0, 2}$ can be implemented with $O(n d^{3})$ 2-qudit gates according to Lemma~\ref{thm:lem:cnx}.

    \bibliographystyle{IEEEtran}

\end{document}